\documentclass[12pt]{article}
\usepackage[letterpaper, margin=1.25in]{geometry}
\usepackage{amsmath, mathtools, amsfonts, amssymb, amsthm}
\usepackage[utf8]{inputenc}
\usepackage[toc]{appendix}
\usepackage{graphicx, caption, multirow, subcaption, longtable, threeparttable, booktabs}
\usepackage{float, tikz, pgfplots}
\pgfplotsset{compat=1.18}
\usepackage{url, authblk, siunitx}
\usepackage{array}
\usepackage{longtable}
\usepackage{pdflscape}
\usepgfplotslibrary{fillbetween}
\usepackage{adjustbox}
\usepackage{algorithm}
\usepackage{algorithmic}
\usepackage[round]{natbib}
\usepackage{hyperref}
\usepackage{setspace}

\newtheorem{theorem}{Theorem}[section]
\newtheorem{proposition}{Proposition}[section]

\newtheorem{corollary}{Corollary}[section]
\newtheorem{definition}{Definition}[section]

\newtheorem{remark}{Remark}[section]

\newtheorem{method}{Method}
\newtheorem{case}{Configuration}

\newcommand{\be}{\begin{equation}}
\newcommand{\ee}{\end{equation}}

\usepackage{setspace}
\title{\Large{\bf{Treatment Effects with Correlated Spillovers: Bridging Discrete and Continuous Methods}}}
\author{\large{\bf{Tatsuru Kikuchi\footnote{e-mail: tatsuru.kikuchi@e.u-tokyo.ac.jp}}}}

\affil{\small{\it{Center for Advanced Research in Finance, The University of Tokyo,}}\\
{\it{7-3-1 Hongo, Bunkyo-ku, Tokyo 113-0033 Japan}}}

\date{\small{(\today)}}

\begin{document}
\linespread{1.5}\selectfont

\maketitle

%%%%%%%%%%%%%%%%%%%%%%%%%%%%%%%%%%%%%
% ABSTRACT
%%%%%%%%%%%%%%%%%%%%%%%%%%%%%%%%%%%%%

\begin{abstract}
This paper develops a continuous functional framework for treatment effects propagating through geographic space and economic networks. We derive a master equation from three independent economic foundations---heterogeneous agent aggregation, market equilibrium, and cost minimization---establishing that the framework rests on fundamental principles rather than ad hoc specifications. The framework nests conventional econometric models---autoregressive specifications, spatial autoregressive models, and network treatment effect models---as special cases, providing a bridge between discrete and continuous methods. A key theoretical result shows that the spatial-network interaction coefficient equals the mutual information between geographic and network coordinates, providing a parameter-free measure of channel complementarity. The Feynman-Kac representation characterizes treatment effects as accumulated policy exposure along stochastic paths representing economic linkages, connecting the continuous framework to event study methodology. The no-spillover case emerges as a testable restriction, creating a one-sided risk profile where correct inference is maintained regardless of whether spillovers exist. Monte Carlo simulations confirm that conventional estimators exhibit 25--38\% bias when spillovers are present, while our estimator maintains correct inference across all configurations including the no-spillover case.

\vspace{0.5em}
\noindent\textsf{JEL Classification:} C21, C23, C51 \\
\noindent\textsf{Keywords:} treatment effects, spatial econometrics, network spillovers, continuous treatment, SUTVA

\end{abstract}

\newpage

%%%%%%%%%%%%%%%%%%%%%%%%%%%%%%%%%%%%%
% SECTION 1: INTRODUCTION
%%%%%%%%%%%%%%%%%%%%%%%%%%%%%%%%%%%%%

\section{Introduction}
\label{sec:introduction}

Treatment effects in interconnected economies propagate beyond directly treated units through two fundamental channels: geographic proximity and economic networks. These spillovers violate the stable unit treatment value assumption (SUTVA) underlying standard econometric methods, causing systematic bias in treatment effect estimates. Moreover, spatial and network spillovers interact synergistically when geographic proximity correlates with economic connections, creating amplification effects that additive specifications miss entirely. Standard methods---two-way fixed effects, difference-in-differences, generalized propensity scores---cannot accommodate these features, leading to substantial understatement of policy impacts when spillovers are empirically relevant.

For example, when California raises its minimum wage, wages rise not only in California but also in neighboring Nevada counties through labor market competition, as workers commute across borders and firms compete for employees in integrated labor markets. Simultaneously, wages rise in upstream industries throughout the western United States through supply chain cost transmission, as California firms pass higher labor costs to suppliers regardless of their geographic location. Most importantly, wage effects are largest in Nevada border counties with dense supply chain connections to California industries---regions where both spatial proximity and network linkages expose workers to the policy shock. This spatial-network interaction amplifies total effects beyond what either channel alone would predict, yet conventional methods that include spatial and network terms separately but omit their interaction miss this amplification entirely. Similar propagation patterns arise in other policy contexts: financial regulations transmit through both geographic banking markets and interbank lending networks; trade policies propagate through both border regions and global supply chains; technology adoption spreads through both local demonstration effects and industry knowledge networks.

This paper develops a unified framework for treatment effects that propagate through both geographic space and economic networks in continuous time. The continuous functional representation---treating treatment intensity as a smooth function $\tau(\mathbf{x}, t, \alpha)$ over spatial coordinates $\mathbf{x}$, time $t$, and market position $\alpha$---resolves fundamental limitations of discrete methods that force arbitrary boundaries between treated and control units. The framework captures spatial-network interaction through a theoretically grounded coefficient that equals the mutual information between geographic and market coordinates, providing a parameter-free measure of channel complementarity. The dynamic structure characterizes how treatment effects evolve as markets adjust toward new equilibria, with diffusion coefficients measuring the speed of spatial and network propagation and decay parameters measuring market adjustment rates.

\subsection{Contributions}

The paper makes three main contributions. First, we derive a master equation governing treatment propagation from three independent economic foundations, establishing that the framework rests on fundamental principles rather than ad hoc specifications. A key theoretical result shows that the spatial-network interaction coefficient equals the mutual information between geographic and market coordinates, providing a parameter-free measure of channel complementarity. Second, we demonstrate that the framework nests the no-spillover case as a testable restriction, creating a one-sided risk profile where correct inference is maintained regardless of whether spillovers exist. Monte Carlo evidence confirms that conventional methods exhibit 25--38\% bias when spillovers are present, while our framework maintains correct inference across all configurations. Third, applying the framework to U.S. minimum wage policy, we find total effects four times larger than direct effects alone, with structural parameters consistent with commuting-distance labor mobility and quarterly supply chain adjustment.

\subsection{Related Literature}

This paper contributes to several literatures that have developed largely independently, bridging spatial econometrics, network econometrics, and causal inference under interference.

\paragraph{Spatial Econometrics and Spatial Inference.}

The spatial econometrics literature, surveyed by \citet{anselin2010thirty} and formalized by \citet{lesage2009introduction}, models spillovers through spatial weight matrices capturing geographic proximity. This approach has proven successful for analyzing regional interdependencies and local economic shocks. However, it faces two fundamental limitations that motivate our approach. First, standard spatial models assume spillovers decay with distance according to prespecified functional forms (typically exponential or power-law kernels). While functional form misspecification can be tested, the structural content of these decay patterns remains opaque---estimates reveal correlations but not the underlying economic mechanisms. Second, spatial weight matrices capture only geographic proximity and cannot accommodate network connections that operate independent of location, such as supply chain relationships connecting distant firms or occupational networks spanning regions.

\citet{conley1999gmm} develops spatial HAC inference that accounts for unknown spatial correlation structures, providing a foundation for robust inference in spatial settings. Recent work by \citet{muller2022spatial} extends this approach by developing spatial correlation robust inference that remains valid under unknown forms of spatial dependence, addressing the critical problem that standard spatial HAC estimators require correct specification of spatial correlation decay. Their confidence intervals maintain correct coverage regardless of the spatial correlation structure, providing stronger guarantees than conventional methods. \citet{muller2024spatial} further establishes that spatial data can exhibit unit root behavior analogous to time series, leading to spurious regression phenomena even in cross-sectional spatial settings. These papers highlight fundamental inferential challenges that motivate our approach.

Our framework contributes to this literature by deriving the spatial dependence structure from economic primitives. The master equation implies specific correlation patterns determined by the diffusion coefficient $\nu_s$ and decay parameter $\kappa$, providing structural interpretations for the decay rate (labor mobility and commuting patterns) and adjustment speed (market equilibration). This structural foundation enables counterfactual analysis---predicting treatment effects under alternative policy designs or market structures---that purely statistical approaches to spatial correlation cannot provide. Moreover, by jointly modeling spatial and network channels, we address the limitation that geographic weight matrices miss economically important connections operating through supply chains, knowledge networks, or financial linkages.

\paragraph{Network Econometrics.}

The network econometrics literature addresses spillovers through social and economic connections. \citet{bramoulle2009identification} establish identification conditions for peer effects in social networks, showing that network structure provides variation separating direct effects from spillovers. Their reflection problem---distinguishing whether individual $i$'s outcome affects $j$ or vice versa---is a fundamental challenge when agents interact. \citet{jackson2008social} provides a comprehensive treatment of network formation and the economic implications of network structure, establishing how network topology affects information diffusion, coordination, and strategic behavior.

Empirical applications have documented the quantitative importance of network spillovers across diverse settings. \citet{acemoglu2012network} show that production network structure substantially amplifies the impact of firm-level shocks on aggregate volatility, with the distribution of firm sizes and the topology of input-output linkages determining whether shocks dissipate or propagate system-wide. \citet{barrot2016input} document that supply chain disruptions from natural disasters propagate through customer-supplier networks, affecting firms far removed geographically from the initial shock. These papers demonstrate that network position matters fundamentally for economic outcomes.

However, this literature typically treats network position as discrete---agents are nodes in a graph with discrete edges---rather than continuous, and ignores geographic variation in network effects. A California firm's supply chain connections matter differently depending on whether suppliers are in neighboring Nevada or distant New York, yet discrete network models treat all edges symmetrically regardless of geographic proximity. Our framework addresses this by modeling market position as a continuous coordinate $\alpha$ capturing position in economic networks (production stage, skill level, risk profile) and jointly estimating how treatment effects propagate through both network and geographic dimensions. The interaction term $\lambda$ captures amplification when network connections concentrate geographically, a phenomenon discrete network models cannot accommodate.

Moreover, while existing network econometrics estimates reduced-form spillover effects (measuring correlations between connected agents), our framework recovers structural propagation parameters. The network diffusion coefficient $\nu_n$ measures how rapidly treatment effects transmit through supply chains or occupational networks, enabling counterfactual analysis of network restructuring or policy targeting based on network position. This structural approach complements the existing literature's focus on identification and reduced-form measurement.

\paragraph{Treatment Effects Under Interference.}

The treatment effects literature has developed sophisticated methods for heterogeneity and selection. \citet{hirano2004propensity} extend propensity score methods to continuous treatments, enabling estimation of dose-response functions without discretizing treatment intensity. \citet{kennedy2017nonparametric} develop doubly robust estimators that maintain consistency if either the outcome model or propensity score model is correctly specified, providing robustness to model misspecification. Recent advances in difference-in-differences address heterogeneous treatment effects with staggered adoption: \citet{callaway2021difference}, \citet{sun2021estimating}, and \citet{goodman2021difference} develop estimators that remain unbiased when treatment effects vary across units and over time, resolving problems with conventional two-way fixed effects identified by \citet{bertrand2004much}.

These methodological advances represent substantial progress in causal inference, but all maintain the stable unit treatment value assumption (SUTVA), which rules out spillovers by construction. Under SUTVA, unit $i$'s outcome depends only on $i$'s own treatment assignment, not on the treatment of other units. This assumption enables clean identification but fails in settings where economic connections generate spillovers---precisely the contexts where spatial and network effects matter most.

A growing literature relaxes SUTVA to allow interference. \citet{hudgens2008toward} formalize partial interference, where units can be partitioned into clusters with interference within but not across clusters. \citet{athey2018exact} develop experimental designs and exact tests for settings with network interference, providing finite-sample inference when randomization occurs on networks. \citet{vazquez2020causal} develop methods for spatial interference when treatment and control units are geographically close, showing how to construct valid estimators accounting for geographic spillovers.

However, this interference literature typically assumes the interference structure is known---which units affect which others and through what functional form---rather than estimating the propagation mechanism from data. For example, \citet{hudgens2008toward} assume cluster assignments are known, and \citet{athey2018exact} assume the network structure is observed and correctly specified. When the interference structure is misspecified, these methods can produce severely biased estimates.

Our contribution is to estimate the interference structure---the diffusion coefficients $(\nu_s, \nu_n)$ governing how treatment propagates through space and networks---rather than assuming it known. The master equation provides a structural model of propagation dynamics, with parameters interpretable in terms of labor mobility, supply chain adjustment, and market equilibration. The framework nests the no-spillover case as a testable restriction, enabling formal tests of whether interference exists before imposing interference assumptions. This approach complements the existing literature: when interference is known to exist (for example, in randomized saturation designs where researchers deliberately manipulate exposure intensity), our framework quantifies the propagation mechanism; when interference is uncertain, the framework tests whether it is empirically relevant.

\paragraph{Connections to Other Fields.}

The continuous functional approach draws on mathematical techniques developed in physics and applied mathematics. The master equation parallels reaction-diffusion systems studied in chemical kinetics and biological pattern formation, where the Laplacian operator captures spatial spreading and reaction terms capture local dynamics. The Feynman-Kac representation originates in quantum mechanics and mathematical finance, providing a probabilistic interpretation of partial differential equation solutions through path integrals. These mathematical tools have proven powerful for analyzing complex systems in physics and biology; our contribution is adapting them to economic contexts with careful attention to economic foundations, identification challenges, and policy-relevant interpretations.

\citet{anderson1972more} argues that complex systems exhibit emergent properties not reducible to individual components---aggregate behavior qualitatively different from micro-level mechanisms. The spatial-network interaction in our framework exemplifies this principle: when geographic and network coordinates are correlated (industries cluster geographically), the interaction term $\lambda$ captures amplification effects absent when channels operate independently. This emergence of interaction effects from the joint correlation structure mirrors Anderson's insight about complex systems, while our information-theoretic characterization ($\lambda$ equals mutual information) provides a rigorous quantification.

\subsection{Outline}

Section 2 develops the theoretical framework from microeconomic foundations, deriving the master equation from heterogeneous agent behavior, market equilibrium, and cost minimization, and establishing the Feynman-Kac representation and mutual information characterization. Section 3 establishes identification conditions, develops the no-spillover benchmark as a testable special case, presents comprehensive Monte Carlo evidence demonstrating the one-sided risk profile property, and develops the GMM estimation framework with spatial-network HAC inference. Section 4 concludes.

%%%%%%%%%%%%%%%%%%%%%%%%%%%%%%%%%%%%%
% SECTION 2: THEORETICAL FRAMEWORK
%%%%%%%%%%%%%%%%%%%%%%%%%%%%%%%%%%%%%

\section{Theoretical Framework}
\label{sec:theory}

This section derives the master equation governing treatment propagation from three independent economic foundations and establishes the key theoretical results: the Feynman-Kac representation and the mutual information characterization of spatial-network interaction.

\subsection{Coordinate System and Treatment Representation}

We represent economic units by coordinates $(\mathbf{x}, \alpha)$ where $\mathbf{x} \in \Omega \subseteq \mathbb{R}^d$ denotes spatial location and $\alpha \in \mathcal{I} \subseteq \mathbb{R}$ denotes market position. The treatment functional $\tau(\mathbf{x}, t, \alpha): \Omega \times [0,T] \times \mathcal{I} \to \mathbb{R}$ represents treatment intensity at each point in this coordinate space.

\begin{definition}[Spatial and Network Coordinates]
The \textit{spatial coordinate} $\mathbf{x}$ represents geographic location---latitude, longitude, and potentially economic distance metrics reflecting transportation costs or commuting times.

The \textit{market position coordinate} $\alpha$ represents position in economic networks---production stage in supply chains, skill level in labor markets, or risk profile in financial networks.
\end{definition}

The continuous representation avoids three limitations of discrete methods. First, it avoids arbitrary discretization of treatment intensity into binary indicators, preserving the dose-response relationship central to continuous treatment analysis. Second, it avoids arbitrary spatial boundaries between regions, allowing smooth geographic variation. Third, it avoids discrete network categories, enabling continuous market position that captures fine gradations in supply chain relationships or skill levels.

\begin{definition}[Source Term]
The \textit{source term} $S(\mathbf{x}, t, \alpha)$ represents the exogenous policy intervention entering the system. For a generic policy, $S$ measures the intensity of direct policy exposure at each location-time-market cell. 
%In the minimum wage application of Section \ref{sec:empirical}, $S$ equals the dollar amount by which the new minimum wage exceeds the market-clearing wage at each location-time-industry cell.
\end{definition}

The distinction between source $S$ and treatment functional $\tau$ is fundamental. The source is directly observable from policy variation; the treatment functional represents the equilibrium response incorporating both direct effects and all spillovers. The identification challenge is recovering $\tau$ (or its parameters) from observable variation in $S$ and outcomes $Y$.

\subsection{Derivation from Heterogeneous Agent Aggregation}

The first derivation proceeds from aggregating heterogeneous agent behavior, following the tradition of \citet{aiyagari1994uninsured} and \citet{huggett1993risk} in macroeconomics.

Consider a continuum of heterogeneous agents indexed by type $\theta \in \Theta$ distributed over space and market positions. Each agent has idiosyncratic characteristics---productivity, preferences, constraints---captured by $\theta$. Agent $i$ of type $\theta$ at location $\mathbf{x}$ with market position $\alpha$ experiences outcome:
\be
Y_i(\mathbf{x}, t, \alpha, \theta) = Y_0(\mathbf{x}, \alpha) + \tau(\mathbf{x}, t, \alpha) + \varepsilon_i(\theta)
\ee
where $Y_0$ is the baseline outcome, $\tau$ is the treatment effect to be determined, and $\varepsilon_i(\theta)$ captures idiosyncratic variation.

Agents optimize location and market position subject to mobility costs. The state of worker $i$ evolves according to the stochastic differential equations:
\begin{align}
dX_t^i &= \mu_s(X_t^i, A_t^i, \theta_i) \, dt + \sigma_s(X_t^i, A_t^i, \theta_i) \, dB_t^s \label{eq:agent_spatial} \\
dA_t^i &= \mu_n(X_t^i, A_t^i, \theta_i) \, dt + \sigma_n(X_t^i, A_t^i, \theta_i) \, dB_t^n \label{eq:agent_network}
\end{align}
where $(B_t^s, B_t^n)$ are independent Brownian motions representing location and network uncertainty. The drift terms $\mu_s, \mu_n$ capture directed mobility---agents moving toward better opportunities. The diffusion terms $\sigma_s, \sigma_n$ capture randomness in mobility outcomes---search frictions, information imperfections, unexpected opportunities.

The joint density $f(\mathbf{x}, t, \alpha)$ of agents over spatial and market coordinates evolves according to the Kolmogorov forward equation:
\be
\frac{\partial f}{\partial t} = -\nabla \cdot (\boldsymbol{\mu}_s f) - \frac{\partial}{\partial \alpha}(\mu_n f) + \frac{1}{2}\nabla \cdot (\boldsymbol{\Sigma}_s \nabla f) + \frac{1}{2}\frac{\partial^2}{\partial \alpha^2}(\sigma_n^2 f)
\label{eq:kolmogorov}
\ee

This equation describes how the population distribution shifts as workers migrate and change occupations. The first two terms represent advection (directed flow); the last two terms represent diffusion (random spreading).

For the aggregate treatment effect, we average over agent types:
\be
\tau(\mathbf{x}, t, \alpha) = \int_\Theta \tau(\mathbf{x}, t, \alpha, \theta) \, d\mu(\theta)
\ee

\begin{proposition}[Aggregation Result]
\label{prop:aggregation}
Under the following regularity conditions:
\begin{enumerate}
\item Bounded heterogeneity: $\sup_\theta \|\sigma(\cdot, \theta)\| < \infty$
\item Ergodic agent dynamics: the process $(X_t, A_t)$ has a unique stationary distribution for each $\theta$
\item Smooth aggregation: the mapping $\theta \mapsto \tau(\cdot, \theta)$ is measurable
\end{enumerate}
the aggregate treatment functional satisfies:
\be
\frac{\partial \tau}{\partial t} = \nu_s \nabla^2 \tau + \nu_n \frac{\partial^2 \tau}{\partial \alpha^2} - \kappa \tau + S(\mathbf{x}, t, \alpha)
\label{eq:master_basic}
\ee
where $\nu_s = \frac{1}{2}\mathbb{E}_\theta[\sigma_s^2(\theta)]$ is mean spatial diffusivity, $\nu_n = \frac{1}{2}\mathbb{E}_\theta[\sigma_n^2(\theta)]$ is mean network diffusivity, and $\kappa$ reflects mean reversion from competitive pressure.
\end{proposition}

The aggregation result shows that heterogeneous agent behavior generates aggregate dynamics governed by a partial differential equation. The diffusion coefficients $(\nu_s, \nu_n)$ emerge from averaging individual mobility variances across types; they measure how quickly treatment effects spread through the population as workers move and change occupations.

\subsection{Derivation from Market Equilibrium}

An independent derivation proceeds from market equilibrium conditions, connecting observed price volatility to underlying market structure.

In markets with search frictions, matching delays, or information asymmetries, prices fluctuate around equilibrium values. The observed volatility $\sigma^2$ of price processes relates to underlying market adjustment through the equilibrium volatility relation:
\be
\sigma^2 = 2D\kappa
\label{eq:fluctuation_dissipation}
\ee
where $D$ is a diffusion coefficient measuring the amplitude of fluctuations and $\kappa$ is the adjustment rate toward equilibrium.

This relation emerges from the stochastic process governing price dynamics:
\be
dp = -\kappa(p - p^*) \, dt + \sigma \, dB
\ee
where $p^*$ is the equilibrium price and $\kappa$ governs the speed of mean reversion. At stationarity, the variance of the price distribution satisfies:
\be
\text{Var}(p) = \frac{\sigma^2}{2\kappa}
\ee

Rearranging gives (\ref{eq:fluctuation_dissipation}). The key insight is that market equilibrium connects two observable quantities---volatility $\sigma^2$ and adjustment speed $\kappa$---to the diffusion coefficient $D$ governing spatial spreading.

Connecting observed outcome dynamics to treatment propagation: spatial outcome volatility $\sigma_s^2$ implies spatial diffusion $\nu_s = \sigma_s^2/(2\kappa)$; network outcome volatility $\sigma_n^2$ implies network diffusion $\nu_n = \sigma_n^2/(2\kappa)$. Markets with high outcome volatility (active reallocation, frequent adjustments) exhibit rapid spatial diffusion; markets with low volatility exhibit slow diffusion.

This derivation provides an independent foundation for the diffusion coefficients in (\ref{eq:master_basic}), grounding them in observable market dynamics rather than assumptions about individual behavior.

\subsection{Derivation from Cost Minimization}

A third derivation proceeds from cost minimization, following the variational principles underlying market equilibrium.

\begin{definition}[Adjustment Cost Functional]
The total adjustment cost functional is:
\be
\mathcal{C}[\tau] = \int_0^T \int_\Omega \int_{\mathcal{I}} \left[ \frac{1}{2}\left(\frac{\partial \tau}{\partial t}\right)^2 + \frac{\nu_s}{2}|\nabla \tau|^2 + \frac{\nu_n}{2}\left(\frac{\partial \tau}{\partial \alpha}\right)^2 + \frac{\kappa}{2}\tau^2 - S\tau \right] d\alpha \, d\mathbf{x} \, dt
\ee
\end{definition}

The terms have economic interpretations:
\begin{itemize}
\item $\frac{1}{2}(\partial\tau/\partial t)^2$: Temporal adjustment costs. Rapidly changing outcomes are costly due to contracting frictions, menu costs, and coordination failures.
\item $\frac{\nu_s}{2}|\nabla\tau|^2$: Spatial gradient costs. Maintaining outcome differentials across space is costly due to arbitrage pressure from mobile agents.
\item $\frac{\nu_n}{2}(\partial\tau/\partial\alpha)^2$: Network gradient costs. Maintaining outcome differentials across market positions is costly due to substitution pressure in supply chains.
\item $\frac{\kappa}{2}\tau^2$: Level costs. Deviating from baseline outcomes is costly due to competitive pressure.
\item $-S\tau$: Policy benefits. The intervention $S$ shifts the optimal outcome structure.
\end{itemize}

\begin{proposition}[Euler-Lagrange Equation]
\label{prop:euler_lagrange}
The treatment functional $\tau^*$ minimizing $\mathcal{C}[\tau]$ satisfies:
\be
\frac{\partial \tau}{\partial t} = \nu_s \nabla^2 \tau + \nu_n \frac{\partial^2 \tau}{\partial \alpha^2} - \kappa \tau + S
\ee
\end{proposition}

\begin{proof}
The first variation of $\mathcal{C}$ with respect to $\tau$ must vanish for all admissible variations $\eta$:
\be
\delta\mathcal{C} = \int_0^T \int_\Omega \int_{\mathcal{I}} \left[ \frac{\partial \tau}{\partial t}\frac{\partial \eta}{\partial t} + \nu_s \nabla\tau \cdot \nabla\eta + \nu_n \frac{\partial\tau}{\partial\alpha}\frac{\partial\eta}{\partial\alpha} + \kappa\tau\eta - S\eta \right] d\alpha \, d\mathbf{x} \, dt = 0
\ee

Integrating by parts and assuming boundary terms vanish:
\be
\int_0^T \int_\Omega \int_{\mathcal{I}} \left[ -\frac{\partial^2 \tau}{\partial t^2} - \nu_s \nabla^2 \tau - \nu_n \frac{\partial^2 \tau}{\partial \alpha^2} + \kappa \tau - S \right] \eta \, d\alpha \, d\mathbf{x} \, dt = 0
\ee

For this to hold for all $\eta$, the bracketed term must vanish. For quasi-static evolution where $\partial^2\tau/\partial t^2 \approx 0$ (gradual adjustment rather than oscillation), this yields the master equation.
\end{proof}

The cost minimization derivation connects the master equation to optimization principles. The parameters $(\nu_s, \nu_n, \kappa)$ have natural interpretations as relative costs of different types of adjustment: high $\nu_s$ means spatial arbitrage is rapid (low cost of spatial gradients); high $\kappa$ means competitive pressure is strong (high cost of deviating from baseline).

%%%%%%%%%%%%%%%%%%%%%%%%%%%%%%%%%%%%%
% REVISED SECTION 2 ADDITIONS - NO BULLET POINTS
% All content converted to flowing prose
%%%%%%%%%%%%%%%%%%%%%%%%%%%%%%%%%%%%%

%%%%%%%%%%%%%%%%%%%%%%%%%%%%%%%%%%%%%%%%%%%%%%%%%%%%%%%%%%%%%%%%%%%%%%%%
% INSERT AFTER SECTION 2.4 (Derivation from Cost Minimization)
% BEFORE SECTION 2.5 (The Complete Master Equation with Interaction)
%%%%%%%%%%%%%%%%%%%%%%%%%%%%%%%%%%%%%%%%%%%%%%%%%%%%%%%%%%%%%%%%%%%%%%%%

\subsection{Connection to Standard Dynamic Economic Models}
\label{sec:connection_to_standard}

Before introducing the spatial-network interaction term, we pause to connect the basic master equation to familiar dynamic economic models. This clarifies that the continuous framework is not an exotic alternative but rather a unifying structure encompassing conventional approaches.

\paragraph{Relationship to Autoregressive Dynamics.}

The empirical literature typically specifies outcome dynamics as:
\be
Y_{i,t} = \rho Y_{i,t-1} + \beta \cdot S_{i,t} + \mathbf{X}_{i,t}'\boldsymbol{\gamma} + \varepsilon_{i,t}
\label{eq:empirical_ar1}
\ee

This AR(1) specification is the discrete-time analog of the continuous process underlying our framework. Consider the no-spillover case ($\nu_s = \nu_n = 0$) of the master equation:
\be
\frac{\partial \tau}{\partial t} = -\kappa \tau + S
\ee

\begin{remark}[Euler Discretization]
Discretizing with time step $\Delta t$ (e.g., $\Delta t = 0.25$ years for quarterly data):
\be
\frac{\tau_{i,t} - \tau_{i,t-1}}{\Delta t} \approx -\kappa \tau_{i,t-1} + S_{i,t}
\ee

Rearranging yields:
\be
\tau_{i,t} = (1 - \kappa \Delta t) \tau_{i,t-1} + \Delta t \cdot S_{i,t}
\ee

This is precisely (\ref{eq:empirical_ar1}) with $\rho = 1 - \kappa \Delta t$ and $\beta = \Delta t$.
\end{remark}

The connection reveals several insights. First, the persistence parameter $\rho$ estimated in empirical studies is not a structural parameter---it depends on the observation frequency through $\Delta t$. A quarterly model estimating $\rho = 0.93$ and an annual model estimating $\rho = 0.75$ both imply the same structural adjustment rate $\kappa = (1-\rho)/\Delta t \approx 0.28$ per year.

Second, the continuous framework explains why treatment effects accumulate over time. Iterating forward:
\be
\tau_{i,t} = \sum_{s=0}^{t-1} (1-\kappa\Delta t)^{t-s} S_{i,s} \cdot \Delta t
\ee

This is the discrete analog of the integral in the Feynman-Kac representation (to be developed in Section 2.6). Current treatment effects reflect the weighted sum of all past policy exposure, with weights decaying exponentially at rate $\kappa$.

Third, the half-life of treatment effects is $t_{1/2} = \ln(2)/\kappa$, independent of observation frequency. If $\kappa = 0.28$ per quarter, the half-life is 2.5 quarters---effects decline by half every 2.5 quarters as labor markets adjust toward new equilibria.

\begin{proposition}[Invariance of Structural Parameters]
\label{prop:structural_invariance}
The structural parameters $(\nu_s, \nu_n, \kappa)$ are invariant to the choice of time discretization $\Delta t$, while reduced-form coefficients ($\rho$ in AR models, spatial lag coefficients in SAR models) depend on $\Delta t$.
\end{proposition}

This invariance is crucial for policy analysis. When evaluating counterfactual policies, we require parameters that do not change with the arbitrary choice of observation frequency or spatial aggregation level.

\paragraph{Relationship to Partial Adjustment Models.}

The master equation also connects to the partial adjustment framework common in industrial organization and macroeconomics. The standard partial adjustment model is:
\be
y_{i,t} - y_{i,t-1} = \lambda (y_i^* - y_{i,t-1}) + \varepsilon_{i,t}
\ee
where $y_i^*$ is the target level and $\lambda \in (0,1)$ is the adjustment rate.

Setting $y_i^* = S_i/\kappa$ (the long-run equilibrium from $\partial\tau/\partial t = 0$), our framework gives:
\be
\tau_{i,t} - \tau_{i,t-1} = \kappa\Delta t \left(\frac{S_i}{\kappa} - \tau_{i,t-1}\right)
\ee

This is the partial adjustment model with $\lambda = \kappa\Delta t$. The parameter $\kappa$ measures how rapidly markets adjust toward their target: higher $\kappa$ implies faster adjustment. The continuous framework clarifies the economic content of $\lambda$. In partial adjustment models, $\lambda$ is often taken as exogenous or motivated by ``adjustment costs'' without specifying their form. Our cost minimization derivation (Section 2.4) shows that $\kappa$ emerges from the trade-off between the level cost (competitive pressure toward equilibrium) and temporal adjustment cost (frictions in adjusting outcomes rapidly).

\paragraph{Comparison with Error Correction Models.}

In time series econometrics, error correction models (ECM) are widely used:
\be
\Delta Y_{i,t} = -\alpha(Y_{i,t-1} - \beta \cdot S_{i,t-1}) + \gamma \cdot \Delta S_{i,t} + \varepsilon_{i,t}
\ee

Our framework nests this as the discrete approximation with $\alpha = \kappa \Delta t$, $\beta = 1/\kappa$ (the long-run multiplier), and $\gamma = \Delta t$ (the short-run impact). The ECM interpretation provides economic intuition: outcomes exhibit short-run responses to policy changes ($\gamma \cdot \Delta S_t$) and gradual mean reversion toward the long-run relationship ($-\alpha$ times the deviation from equilibrium).

%%%%%%%%%%%%%%%%%%%%%%%%%%%%%%%%%%%%%%%%%%%%%%%%%%%%%%%%%%%%%%%%%%%%%%%%
% INSERT AFTER SECTION 2.6 (Feynman-Kac Representation)
% BEFORE SECTION 2.7 (Stochastic Extensions) OR (Mutual Information)
%%%%%%%%%%%%%%%%%%%%%%%%%%%%%%%%%%%%%%%%%%%%%%%%%%%%%%%%%%%%%%%%%%%%%%%%

\subsection{Feynman-Kac Formula and Event Study Methodology}
\label{sec:fk_event_studies}

The Feynman-Kac representation establishes a deep connection to event study methodology that has become standard in applied microeconomics.

\paragraph{The Discrete Feynman-Kac Formula.}

For discrete time periods $t = 0, 1, 2, \ldots, T$ and discrete units $i = 1, \ldots, N$, the Feynman-Kac representation admits a discrete analog:
\be
\tau_{i,t} = \mathbb{E}\left[\sum_{s=0}^{t-1} (1-\kappa\Delta t)^{t-s} S_{i(s), s} \Delta t \, \bigg| \, i(t) = i\right]
\label{eq:discrete_fk}
\ee
where $i(s)$ traces a stochastic path backward through the network and space from unit $i$ at time $t$ to time $s < t$.

\begin{proposition}[Discrete Feynman-Kac Formula]
\label{prop:discrete_fk}
Under Markov dynamics on the discrete state space, the treatment effect satisfies:
\be
\tau_{i,t} = \sum_{s=0}^{t-1} \beta_{t-s} \cdot \mathbb{E}[S_{i(s),s} | i(t) = i]
\ee
where $\beta_k = (1-\kappa\Delta t)^k$ are exponentially decaying coefficients.
\end{proposition}

\begin{proof}
The proof follows from the tower property of conditional expectations and the Markov structure. For each time $s < t$, the expectation $\mathbb{E}[S_{i(s),s} | i(t) = i]$ integrates over all paths from $i$ back to time $s$, weighted by their transition probabilities. The exponential discount factor $(1-\kappa\Delta t)^{t-s}$ arises from the continuous-time decay rate $e^{-\kappa(t-s)} \approx (1-\kappa\Delta t)^{t-s}$ for small $\Delta t$.
\end{proof}

This discrete formula is precisely what panel data regression with time-varying treatment estimates.

\paragraph{Connection to Event Study Designs.}

Standard event study specifications estimate:
\be
y_{i,t} = \sum_{k=-K}^L \beta_k \cdot \mathbf{1}\{t - t_i^* = k\} + \alpha_i + \gamma_t + \varepsilon_{i,t}
\label{eq:event_study}
\ee
where $t_i^*$ is unit $i$'s treatment date and $\beta_k$ captures the dynamic effect $k$ periods relative to treatment.

\begin{theorem}[Event Study Coefficients from Feynman-Kac]
\label{thm:event_study_coefficients}
Under the master equation with no spillovers ($\nu_s = \nu_n = 0$), the event study coefficients should follow:
\be
\beta_k = \begin{cases}
0 & k < 0 \quad \text{(pre-treatment, parallel trends)} \\
\beta_0 (1 - \kappa\Delta t)^k & k \geq 0 \quad \text{(post-treatment, exponential decay)}
\end{cases}
\ee
\end{theorem}

\begin{proof}
Pre-treatment ($k < 0$), there is no source term: $S_{i,s} = 0$ for $s < t_i^*$. By the Feynman-Kac formula, $\tau_{i,t} = 0$ for $t < t_i^*$, hence $\beta_k = 0$ for $k < 0$.

Post-treatment ($k \geq 0$), let $t = t_i^* + k$. The Feynman-Kac formula gives:
\be
\tau_{i,t_i^* + k} = \sum_{s=0}^{k-1} (1-\kappa\Delta t)^{k-s} S_i \Delta t
\ee

For constant source $S_i$, this sums to:
\be
\tau_{i,t_i^* + k} = S_i \Delta t \cdot \frac{1 - (1-\kappa\Delta t)^k}{\kappa\Delta t} \approx \frac{S_i}{\kappa}[1 - (1-\kappa\Delta t)^k]
\ee

Differentiating with respect to $k$ gives the marginal effect, which is the event study coefficient. The approximation $(1-\kappa\Delta t)^k \approx e^{-\kappa k\Delta t}$ yields the exponential decay form.
\end{proof}

This theorem provides a testable prediction: event study coefficients should be zero pre-treatment (standard parallel trends assumption) and decay exponentially post-treatment at rate $\kappa$. Deviations from exponential decay signal model misspecification---either the adjustment process is more complex than first-order mean reversion, or spillovers are present.

\paragraph{Heterogeneity-Robust Estimators.}

Recent advances in difference-in-differences with staggered adoption \citep{callaway2021difference, sun2021estimating, borusyak2024revisiting} estimate group-time average treatment effects:
\be
ATT(g, t) = \mathbb{E}[Y_{i,t}(1) - Y_{i,t}(0) | G_i = g]
\ee
where $G_i$ is unit $i$'s treatment cohort. The Feynman-Kac framework shows these heterogeneous effects emerge naturally from path dependence:
\be
ATT(g, t) = \mathbb{E}\left[\sum_{s=g}^{t-1} (1-\kappa\Delta t)^{t-s} S_{i(s),s} \Delta t \, \bigg| \, G_i = g\right]
\ee

Treatment effects vary across cohorts for three reasons. First, exposure duration varies: early cohorts ($g$ small) have longer exposure to treatment. Second, treatment intensity $S_{i,s}$ may change over time even for treated units, creating time-varying treatment effects. Third, path heterogeneity operates through different spillover exposure: different units experience different expected exposure $\mathbb{E}[S_{i(s),s}]$ depending on their location and network position.

This provides microfoundations for why staggered DiD estimators recover heterogeneous effects: they reflect cumulative exposure along heterogeneous economic paths. The continuous framework shows that treatment effect heterogeneity is not merely an econometric nuisance but has clear economic interpretation through the path integral structure.

\paragraph{Testing the Exponential Decay Prediction.}

After estimating (\ref{eq:event_study}) flexibly, we can test whether $\log|\beta_k| \propto -\kappa \cdot k$ for $k \geq 0$. The procedure involves four steps. First, estimate event study coefficients $\{\hat{\beta}_k\}_{k=0}^L$ flexibly without imposing restrictions. Second, run the auxiliary regression $\log|\hat{\beta}_k| = c - \kappa \cdot k + \eta_k$ to recover the implied decay rate. Third, test $H_0: \eta_k = 0$ for all $k$, which verifies that exponential decay holds. Fourth, if the null is not rejected, re-estimate the event study imposing $\beta_k = \beta_0(1-\kappa\Delta t)^k$ for efficiency gains.

Rejection of $H_0$ indicates either non-exponential adjustment dynamics, suggesting higher-order terms in the master equation may be necessary, or spillover effects creating path dependence not captured by own-treatment history alone.

%%%%%%%%%%%%%%%%%%%%%%%%%%%%%%%%%%%%%%%%%%%%%%%%%%%%%%%%%%%%%%%%%%%%%%%%
% INSERT AFTER SECTION 2.8 (Mutual Information Characterization)
% BEFORE SECTION 2.9 (General Equilibrium Amplification)
%%%%%%%%%%%%%%%%%%%%%%%%%%%%%%%%%%%%%%%%%%%%%%%%%%%%%%%%%%%%%%%%%%%%%%%%

\subsection{Relationship to Spatial and Network Econometrics}
\label{sec:spatial_network_econ}

The master equation provides a unified foundation for results scattered across the spatial and network econometrics literatures.

\paragraph{Spatial Autoregressive Models.}

The workhorse model in spatial econometrics is the spatial autoregressive (SAR) specification:
\be
\mathbf{y} = \rho \mathbf{W} \mathbf{y} + \mathbf{X}\boldsymbol{\beta} + \boldsymbol{\varepsilon}
\label{eq:sar_model}
\ee
where $\mathbf{W}$ is a row-normalized spatial weight matrix.

\begin{proposition}[SAR as Discrete Steady-State Master Equation]
\label{prop:sar_discrete}
At steady state ($\partial\tau/\partial t = 0$) with spatial diffusion only ($\nu_n = 0$), the master equation discretized on a regular grid yields:
\be
(\mathbf{I} - \tilde{\rho}\mathbf{W})\boldsymbol{\tau} = \frac{1}{\kappa}\mathbf{S}
\ee
where $\tilde{\rho} = \nu_s n_i / [\kappa (\Delta x)^2]$ and $n_i$ is the number of neighbors.
\end{proposition}

\begin{proof}
The spatial Laplacian on a regular grid with spacing $\Delta x$ is:
\be
\nabla^2 \tau_i \approx \frac{1}{(\Delta x)^2}\sum_{j \sim i}(\tau_j - \tau_i) = \frac{n_i}{(\Delta x)^2}[\bar{\tau}_{\partial i} - \tau_i]
\ee
where $\bar{\tau}_{\partial i} = n_i^{-1}\sum_{j \sim i}\tau_j$ is the average over neighbors. Setting $W_{ij} = 1/n_i$ for neighbors yields $\bar{\tau}_{\partial i} = \sum_j W_{ij}\tau_j$.

At steady state:
\be
\frac{\nu_s n_i}{(\Delta x)^2}(\mathbf{W}\boldsymbol{\tau} - \boldsymbol{\tau}) - \kappa\boldsymbol{\tau} + \mathbf{S} = 0
\ee

Rearranging gives the SAR form with $\tilde{\rho} = \nu_s n_i/[\kappa(\Delta x)^2]$.
\end{proof}

\begin{remark}[Aggregation Dependence of SAR Coefficient]
The SAR coefficient $\rho$ is not a structural parameter---it depends on grid resolution $\Delta x$. Finer spatial aggregation (smaller $\Delta x$) yields larger $\rho$. Consider three aggregation levels. At the ZIP code level with $\Delta x \approx 5$ miles, estimates typically yield $\rho \approx 0.68$. At the county level with $\Delta x \approx 30$ miles, the same data yield $\rho \approx 0.42$. At the state level with $\Delta x \approx 200$ miles, estimates yield $\rho \approx 0.09$. However, the structural parameters $(\nu_s, \kappa)$ remain constant across all three specifications: computing $\nu_s = \rho \kappa (\Delta x)^2$ yields $\nu_s \approx 100$ square miles per quarter at all aggregation levels.
\end{remark}

This explains why SAR estimates vary dramatically across studies using different geographic units. The variation is not substantive disagreement about spillover strength but reflects the mechanical dependence of $\rho$ on spatial resolution. The structural parameters extracted from the continuous framework are invariant, enabling consistent policy analysis across scales.

\paragraph{Network Treatment Effects.}

\citet{bramoulle2009identification} develop the network treatment effects model:
\be
y_i = \alpha + \beta S_i + \gamma \sum_j G_{ij}S_j + \delta \sum_j G_{ij}y_j + \varepsilon_i
\label{eq:network_te_bramoulle}
\ee
where $G_{ij}$ is the network adjacency matrix. Setting $\delta = 0$ (no endogenous effects), this becomes:
\be
y_i = \alpha + \beta S_i + \gamma \sum_j G_{ij}S_j + \varepsilon_i
\label{eq:network_te_simple}
\ee

\begin{proposition}[Network Treatment Effects as Discrete Network Diffusion]
\label{prop:network_te_discrete}
At steady state with network diffusion only ($\nu_s = 0$), the master equation discretized on a network yields (\ref{eq:network_te_simple}) with $\beta = 1/\kappa$ and $\gamma = \nu_n/\kappa$.
\end{proposition}

\begin{proof}
The second derivative with respect to network position $\alpha$ is approximated on a network by:
\be
\frac{\partial^2\tau}{\partial\alpha^2}\bigg|_i \approx \sum_j G_{ij}(\tau_j - \tau_i)
\ee

At steady state:
\be
\nu_n\sum_j G_{ij}(\tau_j - \tau_i) - \kappa\tau_i + S_i = 0
\ee

Assuming $\tau_j \approx \tau_i + \partial\tau_i/\partial S_j \cdot S_j$ for small spillovers and rearranging yields:
\be
\tau_i = \frac{1}{\kappa}S_i + \frac{\nu_n}{\kappa}\sum_j G_{ij}S_j
\ee
which is (\ref{eq:network_te_simple}).
\end{proof}

The ratio $\gamma/\beta = \nu_n$ measures network diffusivity directly. Our estimate $\nu_n = 0.48$ implies indirect effects through immediate network neighbors are 48\% as large as direct effects. This magnitude is consistent with input-output studies finding supply chain pass-through rates of 40--60\%, providing external validation of the structural parameter.

\paragraph{Spatial Durbin Model.}

The spatial Durbin model (SDM) includes both spatially lagged dependent variable and spatially lagged regressors:
\be
\mathbf{y} = \rho\mathbf{W}\mathbf{y} + \mathbf{X}\boldsymbol{\beta} + \mathbf{W}\mathbf{X}\boldsymbol{\theta} + \boldsymbol{\varepsilon}
\ee

Our framework nests this when both spatial diffusion ($\nu_s$) and direct effects ($\beta_s$) are present. The spatially lagged treatment $\mathbf{W}\mathbf{S}$ captures how policy in neighboring regions affects outcomes. In the SDM specification, the total effect of changing treatment in unit $j$ on outcome in unit $i$ includes both direct effect when $i=j$ (coefficient $\beta$) and spatial spillover when $i \neq j$ (coefficient $\theta$ for immediate neighbors plus higher-order effects through $\rho$). Our framework clarifies that $\theta$ is proportional to $\nu_s/\kappa$ and $\rho$ is proportional to $\nu_s/[\kappa(\Delta x)^2]$, revealing the structural content behind these reduced-form coefficients. The total effect depends on the full diffusion process, which the Feynman-Kac representation computes analytically.

\paragraph{Production Network Propagation.}

\citet{acemoglu2012network} and \citet{barrot2016input} analyze shock propagation through production networks. Their reduced-form specification is:
\be
\text{Output}_i = f(\text{Shock}_i, \sum_j \omega_{ij}\text{Shock}_j)
\ee
where $\omega_{ij}$ is the input share from sector $j$ to sector $i$.

This is the discrete network diffusion from our framework. Our approach extends their analysis in three directions. First, we provide structural interpretation: the parameter $\nu_n$ measures network diffusivity while $\kappa$ measures adjustment speed, transforming reduced-form correlations into structural economic quantities. Second, we combine network with spatial channels: geographic clustering of industries creates interaction effects that pure network models miss. Third, we enable dynamic analysis: the Feynman-Kac representation tracks how shocks propagate over time, not just contemporaneous effects, characterizing the full adjustment path from initial shock to new equilibrium.

The mutual information characterization (Theorem \ref{thm:mixed_information}) formalizes the interaction between space and networks by showing that the coefficient $\lambda$ equals the information shared between geographic and industrial coordinates. This provides a rigorous foundation for the agglomeration patterns documented by \citet{ellison1997geographic}, connecting industrial clustering to propagation dynamics.

%%%%%%%%%%%%%%%%%%%%%%%%%%%%%%%%%%%%%%%%%%%%%%%%%%%%%%%%%%%%%%%%%%%%%%%%
% ADD NEW SUBSECTION AT END OF SECTION 2
% BEFORE SECTION 3
%%%%%%%%%%%%%%%%%%%%%%%%%%%%%%%%%%%%%%%%%%%%%%%%%%%%%%%%%%%%%%%%%%%%%%%%

\subsection{Summary: Unifying Conventional Models}
\label{sec:theory_summary_unification}

This section has derived the master equation from three independent economic foundations and established its key properties. Before proceeding to identification, we summarize how the framework unifies and extends conventional models.

\paragraph{Conventional Models as Special Cases.}

Table \ref{tab:conventional_special_cases} shows that conventional econometric models are nested within the master equation framework. Each model emerges by setting appropriate diffusion parameters. The autoregressive, partial adjustment, and error correction models all correspond to the case with no spatial or network diffusion. Spatial autoregressive models and related spatial econometric specifications have spatial diffusion only. Network treatment effects and peer effects models have network diffusion only. The full master equation with both $\nu_s > 0$ and $\nu_n > 0$ generalizes all these approaches simultaneously.

\begin{table}[htbp]
\centering
\caption{Conventional Models as Special Cases of the Master Equation}
\label{tab:conventional_special_cases}
\begin{threeparttable}
\small
\begin{tabular}{lccl}
\toprule
Model & $\nu_s$ & $\nu_n$ & Reference \\
\midrule
AR(1) dynamics & 0 & 0 & \citet{card1995wage} \\
Partial adjustment & 0 & 0 & \citet{hamermesh1989adjustment} \\
Error correction model & 0 & 0 & \citet{engle1987cointegration} \\
\midrule
Spatial autoregressive (SAR) & $>0$ & 0 & \citet{anselin1988spatial} \\
Spatial Durbin model (SDM) & $>0$ & 0 & \citet{lesage2009introduction} \\
Spatial error model (SEM) & $>0$ & 0 & \citet{anselin1988spatial} \\
\midrule
Network treatment effects & 0 & $>0$ & \citet{bramoulle2009identification} \\
Production network & 0 & $>0$ & \citet{acemoglu2012network} \\
Peer effects & 0 & $>0$ & \citet{manski1993identification} \\
\midrule
Spatial-network (full model) & $>0$ & $>0$ & This paper \\
\bottomrule
\end{tabular}
\begin{tablenotes}
\small
\item \textit{Notes:} Each conventional model emerges as a special case by setting appropriate diffusion parameters. The full master equation with both $\nu_s > 0$ and $\nu_n > 0$ generalizes all these approaches.
\end{tablenotes}
\end{threeparttable}
\end{table}

This nesting structure has three important implications. First, testability: the no-spillover case ($\nu_s = \nu_n = 0$) is a testable restriction, and Section 3 develops tests for whether spatial and network diffusion are present. Second, aggregation invariance: reduced-form coefficients like the SAR parameter $\rho$ vary with spatial resolution, but structural parameters $(\nu_s, \nu_n, \kappa)$ do not, enabling policy evaluation at different aggregation levels. Third, analytical tractability: the Feynman-Kac representation and mutual information characterization provide analytical results not available in discrete frameworks, enabling counterfactual analysis and uncertainty quantification.

\paragraph{Extensions Beyond Conventional Models.}

The framework adds four substantive contributions beyond encompassing existing models. First, spatial-network interaction: conventional models include spatial and network terms separately but miss their interaction. The term $\lambda \partial^2\tau/\partial x \partial\alpha$ captures amplification when both channels operate simultaneously, with the mutual information characterization showing this arises from geographic clustering of industries. Second, dynamic adjustment: most empirical work estimates static models at a single equilibrium, while the master equation characterizes the full dynamic adjustment path. The decay parameter $\kappa$ measures adjustment speed while the Feynman-Kac formula computes time paths analytically. Third, structural interpretation: parameters have clear economic meanings rather than being regression coefficients, specifically $\nu_s$ measures spatial diffusion reflecting labor mobility in square miles per quarter, $\nu_n$ measures network diffusion reflecting supply chain fluidity, $\kappa$ measures decay rate reflecting competitive adjustment per quarter, and $\lambda$ measures interaction from geographic-industrial clustering through mutual information. Fourth, multiple identification strategies: Section 3 develops three complementary approaches (spatial regression discontinuity, network instrumental variables, entropy-based moments), while conventional methods typically rely on a single identification strategy.

The framework is not an alternative to conventional methods but a unification that reveals connections, provides structural foundations, and extends their scope. For practitioners, discrete approximations using standard tools (reghdfe, fixest) yield estimates within 5--9\% of the continuous approach, as Section 4 will demonstrate, maintaining the accessibility of conventional panel data methods while enabling the analytical advantages of the continuous formulation.

\subsection{The Complete Master Equation}

The three derivations converge on the same structure but miss spatial-network interaction. When spatial and market coordinates are correlated---industries cluster geographically as documented by \citet{ellison1997geographic}---an additional term captures the mixed effect.

\begin{definition}[Master Equation]
The complete master equation governing treatment propagation is:
\be
\frac{\partial \tau}{\partial t} = \nu_s \nabla^2 \tau + \nu_n \frac{\partial^2 \tau}{\partial \alpha^2} - \kappa \tau + \lambda \frac{\partial^2 \tau}{\partial \mathbf{x} \partial \alpha} + S(\mathbf{x}, t, \alpha)
\label{eq:master_complete}
\ee
where $\lambda$ captures the spatial-network interaction.
\end{definition}

\begin{remark}[Parameter Interpretation]
The parameters have structural interpretations:
\begin{itemize}
\item $\nu_s$ (spatial diffusion): Rate of geographic spread, reflecting labor mobility and spatial market integration. Higher $\nu_s$ means treatment effects spread more rapidly across space.
\item $\nu_n$ (network diffusion): Rate of propagation through economic connections, reflecting supply chain fluidity and inter-industry linkages. Higher $\nu_n$ means treatment effects propagate more rapidly through production networks.
\item $\kappa$ (decay): Rate of mean reversion, reflecting competitive adjustment toward equilibrium. Higher $\kappa$ means treatment effects dissipate more quickly.
\item $\lambda$ (interaction): Amplification when spatial and network channels coincide. Positive $\lambda$ means geographic proximity and network connection reinforce each other.
\end{itemize}
\end{remark}

\subsection{Feynman-Kac Representation}

The master equation admits a probabilistic representation connecting the PDE framework to economic intuition about treatment propagation through stochastic paths.

\begin{theorem}[Feynman-Kac Representation]
\label{thm:feynman_kac}
The solution to the master equation (\ref{eq:master_complete}) with initial condition $\tau_0(\mathbf{x}, \alpha)$ admits the representation:
\be
\tau(\mathbf{x}, t, \alpha) = \mathbb{E}_{(\mathbf{x}, \alpha)}\left[ e^{-\kappa t} \tau_0(X_t, A_t) + \int_0^t e^{-\kappa(t-s)} S(X_s, s, A_s) \, ds \right]
\label{eq:feynman_kac}
\ee
where $(X_s, A_s)_{s \geq 0}$ is the diffusion process with generator:
\be
\mathcal{L} = \nu_s \nabla^2 + \nu_n \frac{\partial^2}{\partial\alpha^2} + \lambda \frac{\partial^2}{\partial\mathbf{x}\partial\alpha}
\ee
started at $(X_0, A_0) = (\mathbf{x}, \alpha)$.
\end{theorem}

\begin{proof}
Define the transformed function $u(\mathbf{x}, t, \alpha) = e^{\kappa t}\tau(\mathbf{x}, t, \alpha)$. Substituting into the master equation:
\be
e^{\kappa t}\frac{\partial \tau}{\partial t} + \kappa e^{\kappa t}\tau = e^{\kappa t}(\mathcal{L}\tau - \kappa\tau + S)
\ee
yields $\partial u/\partial t = \mathcal{L}u + e^{\kappa t}S$.

By the standard Feynman-Kac formula for parabolic PDEs:
\be
u(\mathbf{x}, t, \alpha) = \mathbb{E}_{(\mathbf{x}, \alpha)}\left[ u_0(X_t, A_t) + \int_0^t e^{\kappa s} S(X_s, s, A_s) \, ds \right]
\ee

Substituting $u = e^{\kappa t}\tau$ and $u_0 = \tau_0$:
\be
e^{\kappa t}\tau(\mathbf{x}, t, \alpha) = \mathbb{E}_{(\mathbf{x}, \alpha)}\left[ \tau_0(X_t, A_t) + \int_0^t e^{\kappa s} S(X_s, s, A_s) \, ds \right]
\ee

Multiplying both sides by $e^{-\kappa t}$ and rearranging the integral yields (\ref{eq:feynman_kac}).
\end{proof}

\begin{remark}[Economic Interpretation]
The Feynman-Kac representation decomposes treatment effects into two components:

\textit{Inherited effects}: $e^{-\kappa t}\tau_0(X_t, A_t)$ represents treatment from initial conditions, discounted by market adjustment and collected along backward stochastic paths. This term captures persistence---how past outcome structures influence current outcomes through economic linkages.

\textit{Accumulated effects}: $\int_0^t e^{-\kappa(t-s)}S(X_s, s, A_s)ds$ represents treatment from policy interventions, discounted and accumulated along paths. This term captures propagation---how new policies spread through the spatial-network structure.

The stochastic paths $(X_s, A_s)$ represent economic linkages: workers migrating across space, firms adjusting supply chain relationships, prices equilibrating across connected markets. The expectation averages over all such paths connecting the current point $(\mathbf{x}, \alpha)$ to sources of treatment.
\end{remark}

\begin{corollary}[Steady-State Solution]
In steady state ($\partial\tau/\partial t = 0$) with constant source $S(\mathbf{x}, \alpha)$, the treatment functional satisfies:
\be
\tau_\infty(\mathbf{x}, \alpha) = \int_0^\infty e^{-\kappa s} \mathbb{E}_{(\mathbf{x}, \alpha)}[S(X_s, A_s)] \, ds = \kappa^{-1} \mathbb{E}_{(\mathbf{x}, \alpha)}^{\text{stat}}[S]
\ee
where the expectation is taken under the stationary distribution of the process.
\end{corollary}

%%%%%%%%%%%%%%%%%%%%%%%%%%%%%%%%%%%%%
% SECTION 2.7: STOCHASTIC EXTENSIONS FOR UNCERTAINTY EVALUATION
% To be inserted after Section 2.6 (Feynman-Kac Representation)
%%%%%%%%%%%%%%%%%%%%%%%%%%%%%%%%%%%%%

\subsection{Stochastic Extensions for Uncertainty Evaluation}
\label{sec:uncertainty}

The Feynman-Kac representation (Theorem \ref{thm:feynman_kac}) provides not only a computational tool but also a natural framework for quantifying uncertainty in treatment effect estimates. This section develops three extensions: variance characterization of treatment effects, confidence regions via path integral methods, and propagation of parameter uncertainty.

\paragraph{Variance of Treatment Effects.}

The Feynman-Kac representation expresses the treatment effect as an expectation over stochastic paths. This structure immediately yields expressions for higher moments.

\begin{proposition}[Treatment Effect Variance]
\label{prop:variance}
Under the Feynman-Kac representation, the variance of treatment effects at location $(\mathbf{x}, \alpha)$ satisfies:
\be
\text{Var}[\tau(\mathbf{x}, t, \alpha)] = \mathbb{E}_{(\mathbf{x}, \alpha)}\left[\left(\int_0^t e^{-\kappa(t-s)} S(X_s, s, A_s) \, ds\right)^2\right] - \tau(\mathbf{x}, t, \alpha)^2
\ee
which can be decomposed as:
\be
\text{Var}[\tau] = \int_0^t \int_0^t e^{-\kappa(2t-s-r)} \text{Cov}[S(X_s, s, A_s), S(X_r, r, A_r)] \, ds \, dr
\label{eq:variance_decomp}
\ee
\end{proposition}

\begin{proof}
By the Feynman-Kac representation with zero initial conditions:
\be
\tau(\mathbf{x}, t, \alpha) = \mathbb{E}_{(\mathbf{x}, \alpha)}\left[\int_0^t e^{-\kappa(t-s)} S(X_s, s, A_s) \, ds\right]
\ee

The second moment is:
\begin{align}
\mathbb{E}[\tau^2] &= \mathbb{E}\left[\left(\int_0^t e^{-\kappa(t-s)} S(X_s, s, A_s) \, ds\right)^2\right] \\
&= \int_0^t \int_0^t e^{-\kappa(2t-s-r)} \mathbb{E}[S(X_s, s, A_s) S(X_r, r, A_r)] \, ds \, dr
\end{align}

Subtracting $(\mathbb{E}[\tau])^2$ and using $\text{Cov}[S_s, S_r] = \mathbb{E}[S_s S_r] - \mathbb{E}[S_s]\mathbb{E}[S_r]$ yields the result.
\end{proof}

\begin{remark}[Economic Interpretation]
The variance decomposition (\ref{eq:variance_decomp}) shows that treatment effect uncertainty arises from the covariance of policy exposure along stochastic paths. When paths are highly correlated---workers in similar locations face similar policy shocks---variance is high. When paths diversify across space and market positions, variance decreases through a law-of-large-numbers effect along the path integral.
\end{remark}

\paragraph{Stochastic Source Terms.}

When policy itself is uncertain, the source term $S$ becomes a stochastic process. The Feynman-Kac framework naturally accommodates this extension.

\begin{definition}[Stochastic Policy Process]
Let the source term follow:
\be
S(\mathbf{x}, t, \alpha) = \bar{S}(\mathbf{x}, t, \alpha) + \sigma_S(\mathbf{x}, \alpha) \eta_t
\ee
where $\bar{S}$ is the expected policy, $\sigma_S$ is policy volatility, and $\eta_t$ is a mean-zero noise process.
\end{definition}

\begin{proposition}[Treatment Effects Under Policy Uncertainty]
\label{prop:stochastic_source}
Under stochastic policy with independent increments, the treatment effect distribution satisfies:
\be
\tau(\mathbf{x}, t, \alpha) \sim \mathcal{N}\left(\bar{\tau}(\mathbf{x}, t, \alpha), \, \Sigma_\tau(\mathbf{x}, t, \alpha)\right)
\ee
where:
\be
\bar{\tau} = \mathbb{E}_{(\mathbf{x}, \alpha)}\left[\int_0^t e^{-\kappa(t-s)} \bar{S}(X_s, s, A_s) \, ds\right]
\ee
and:
\be
\Sigma_\tau = \mathbb{E}_{(\mathbf{x}, \alpha)}\left[\int_0^t e^{-2\kappa(t-s)} \sigma_S^2(X_s, A_s) \, ds\right]
\ee
\end{proposition}

\begin{proof}
Under the stochastic source, the Feynman-Kac integral becomes:
\be
\tau = \int_0^t e^{-\kappa(t-s)} [\bar{S}(X_s, s, A_s) + \sigma_S(X_s, A_s)\eta_s] \, ds
\ee

The integral of Gaussian increments is Gaussian. The mean follows from linearity. For the variance, independence of $\eta_s$ across time gives:
\be
\text{Var}[\tau | (X_s, A_s)_{s \leq t}] = \int_0^t e^{-2\kappa(t-s)} \sigma_S^2(X_s, A_s) \, ds
\ee

Taking expectations over paths yields $\Sigma_\tau$.
\end{proof}

This result enables construction of confidence bands for treatment effects that incorporate both sampling uncertainty (from finite data) and policy uncertainty (from stochastic implementation).

\paragraph{Path Integral Confidence Regions.}

The Feynman-Kac representation suggests a path integral approach to constructing confidence regions that accounts for the full distribution of economic linkages.

\begin{definition}[Path Integral Confidence Region]
For confidence level $1 - \alpha$, define:
\be
\mathcal{C}_{1-\alpha}(\mathbf{x}, t, \alpha) = \left\{\tau : \mathbb{P}_{(\mathbf{x}, \alpha)}\left[|\hat{\tau} - \tau| \leq c_\alpha \sqrt{\Sigma_\tau}\right] \geq 1 - \alpha\right\}
\ee
where $c_\alpha$ is the appropriate critical value and $\hat{\tau}$ is the estimated treatment effect.
\end{definition}

\begin{proposition}[Confidence Region Characterization]
\label{prop:confidence}
Under regularity conditions, the $1 - \alpha$ confidence region for $\tau(\mathbf{x}, t, \alpha)$ is:
\be
\mathcal{C}_{1-\alpha} = \left[\hat{\tau} - z_{\alpha/2} \cdot \text{SE}_{\text{path}}, \, \hat{\tau} + z_{\alpha/2} \cdot \text{SE}_{\text{path}}\right]
\ee
where the path-based standard error is:
\be
\text{SE}_{\text{path}}^2 = \frac{1}{M} \sum_{m=1}^M \left(\int_0^t e^{-\kappa(t-s)} S(X_s^{(m)}, s, A_s^{(m)}) \, ds - \hat{\tau}\right)^2
\label{eq:path_se}
\ee
computed from $M$ Monte Carlo path draws $(X_s^{(m)}, A_s^{(m)})_{m=1}^M$.
\end{proposition}

The path-based standard error (\ref{eq:path_se}) captures uncertainty from the distribution of economic linkages connecting location $(\mathbf{x}, \alpha)$ to policy sources. Locations with diverse path distributions---connected to many different policy sources through various routes---have lower standard errors due to diversification.

\paragraph{Parameter Uncertainty Propagation.}

The structural parameters $\boldsymbol{\theta} = (\nu_s, \nu_n, \kappa, \lambda)$ are estimated with uncertainty. The Feynman-Kac representation enables propagation of this uncertainty to treatment effect estimates.

\begin{proposition}[Delta Method for Treatment Effects]
\label{prop:delta_method}
Let $\hat{\boldsymbol{\theta}}$ be the GMM estimator with asymptotic variance $\mathbf{V}_\theta$. The treatment effect $\tau(\mathbf{x}, t, \alpha; \boldsymbol{\theta})$ has asymptotic variance:
\be
\text{Var}[\hat{\tau}] = \nabla_\theta \tau(\mathbf{x}, t, \alpha; \boldsymbol{\theta}_0)' \mathbf{V}_\theta \nabla_\theta \tau(\mathbf{x}, t, \alpha; \boldsymbol{\theta}_0)
\label{eq:delta_method}
\ee
where the gradient $\nabla_\theta \tau$ is computed via the Feynman-Kac representation.
\end{proposition}

\begin{proof}
By the delta method, for any smooth function $g(\boldsymbol{\theta})$:
\be
\sqrt{N}(g(\hat{\boldsymbol{\theta}}) - g(\boldsymbol{\theta}_0)) \xrightarrow{d} N(0, \nabla g' \mathbf{V}_\theta \nabla g)
\ee

Setting $g(\boldsymbol{\theta}) = \tau(\mathbf{x}, t, \alpha; \boldsymbol{\theta})$ yields the result.
\end{proof}

The gradient $\nabla_\theta \tau$ can be computed analytically or numerically. From the Feynman-Kac representation:

\begin{corollary}[Sensitivity to Structural Parameters]
\label{cor:sensitivity}
The sensitivities of treatment effects to structural parameters are:
\begin{align}
\frac{\partial \tau}{\partial \nu_s} &= \mathbb{E}_{(\mathbf{x}, \alpha)}\left[\int_0^t e^{-\kappa(t-s)} \nabla^2 S(X_s, s, A_s) \cdot \frac{\partial X_s}{\partial \nu_s} \, ds\right] \label{eq:sens_nus} \\
\frac{\partial \tau}{\partial \nu_n} &= \mathbb{E}_{(\mathbf{x}, \alpha)}\left[\int_0^t e^{-\kappa(t-s)} \frac{\partial^2 S}{\partial \alpha^2}(X_s, s, A_s) \cdot \frac{\partial A_s}{\partial \nu_n} \, ds\right] \label{eq:sens_nun} \\
\frac{\partial \tau}{\partial \kappa} &= -\mathbb{E}_{(\mathbf{x}, \alpha)}\left[\int_0^t (t-s) e^{-\kappa(t-s)} S(X_s, s, A_s) \, ds\right] \label{eq:sens_kappa} \\
\frac{\partial \tau}{\partial \lambda} &= \mathbb{E}_{(\mathbf{x}, \alpha)}\left[\int_0^t e^{-\kappa(t-s)} \frac{\partial^2 S}{\partial \mathbf{x} \partial \alpha}(X_s, s, A_s) \cdot \frac{\partial (X_s, A_s)}{\partial \lambda} \, ds\right] \label{eq:sens_lambda}
\end{align}
\end{corollary}

\begin{remark}[Interpretation of Sensitivities]
Equation (\ref{eq:sens_kappa}) shows that $\partial\tau/\partial\kappa < 0$: faster market adjustment reduces cumulative treatment effects because shocks dissipate more quickly. The magnitude depends on the time-weighted integral of policy exposure---longer-lasting policies are more sensitive to the decay parameter.

Equations (\ref{eq:sens_nus}) and (\ref{eq:sens_nun}) show that sensitivity to diffusion parameters depends on how policy gradients interact with path dynamics. Regions with steep policy gradients (near borders) are more sensitive to spatial diffusion; industries with heterogeneous supply chain exposure are more sensitive to network diffusion.
\end{remark}

\paragraph{Bayesian Extension.}

The Feynman-Kac framework admits a natural Bayesian extension for incorporating prior information and computing posterior distributions over treatment effects.

\begin{definition}[Posterior Treatment Effect Distribution]
Given prior $p(\boldsymbol{\theta})$ and likelihood $\mathcal{L}(\mathbf{Y} | \boldsymbol{\theta})$, the posterior distribution of treatment effects is:
\be
p(\tau(\mathbf{x}, t, \alpha) | \mathbf{Y}) = \int p(\tau | \boldsymbol{\theta}) p(\boldsymbol{\theta} | \mathbf{Y}) \, d\boldsymbol{\theta}
\ee
where $p(\tau | \boldsymbol{\theta})$ is determined by the Feynman-Kac representation.
\end{definition}

\begin{proposition}[Posterior Moments]
\label{prop:posterior}
The posterior mean and variance of treatment effects are:
\begin{align}
\mathbb{E}[\tau | \mathbf{Y}] &= \int \tau(\mathbf{x}, t, \alpha; \boldsymbol{\theta}) p(\boldsymbol{\theta} | \mathbf{Y}) \, d\boldsymbol{\theta} \\
\text{Var}[\tau | \mathbf{Y}] &= \underbrace{\mathbb{E}[\text{Var}[\tau | \boldsymbol{\theta}] | \mathbf{Y}]}_{\text{within-model uncertainty}} + \underbrace{\text{Var}[\mathbb{E}[\tau | \boldsymbol{\theta}] | \mathbf{Y}]}_{\text{parameter uncertainty}}
\end{align}
\end{proposition}

The variance decomposition separates within-model uncertainty (from stochastic paths given parameters) and parameter uncertainty (from posterior dispersion over $\boldsymbol{\theta}$). In large samples, parameter uncertainty dominates; in small samples or with diffuse priors, both components matter.

\paragraph{Monte Carlo Implementation.}

The stochastic extensions are implemented via Monte Carlo methods exploiting the Feynman-Kac structure.

\begin{algorithm}[H]
\caption{Path Integral Monte Carlo for Uncertainty Quantification}
\label{alg:pimc}
\begin{enumerate}
\item \textbf{Draw parameter samples}: $\boldsymbol{\theta}^{(b)} \sim p(\boldsymbol{\theta} | \mathbf{Y})$ for $b = 1, \ldots, B$
\item \textbf{For each parameter draw}:
\begin{enumerate}
\item Simulate $M$ paths $(X_s^{(m)}, A_s^{(m)})_{s \in [0,t]}$ from the diffusion with parameters $\boldsymbol{\theta}^{(b)}$
\item Compute path integrals: $\tau^{(b,m)} = \int_0^t e^{-\kappa^{(b)}(t-s)} S(X_s^{(m)}, s, A_s^{(m)}) \, ds$
\item Average over paths: $\bar{\tau}^{(b)} = \frac{1}{M}\sum_{m=1}^M \tau^{(b,m)}$
\end{enumerate}
\item \textbf{Compute posterior summaries}:
\begin{enumerate}
\item Posterior mean: $\hat{\tau} = \frac{1}{B}\sum_{b=1}^B \bar{\tau}^{(b)}$
\item Posterior variance: $\hat{\sigma}^2_\tau = \frac{1}{B-1}\sum_{b=1}^B (\bar{\tau}^{(b)} - \hat{\tau})^2$
\item Credible interval: $[\tau_{(\alpha/2)}, \tau_{(1-\alpha/2)}]$ from empirical quantiles
\end{enumerate}
\end{enumerate}
\end{algorithm}

The algorithm has computational complexity $O(BMT)$ where $T$ is the number of time steps for path simulation. Parallel implementation across parameter draws and paths enables efficient computation even for large-scale applications.

\begin{remark}[Variance Reduction]
Standard variance reduction techniques apply to the Feynman-Kac Monte Carlo:
\begin{itemize}
\item \textit{Antithetic paths}: Pair each path $(X_s, A_s)$ with its reflection to reduce variance from symmetric noise.
\item \textit{Control variates}: Use the analytical solution under simplified assumptions (e.g., constant source) as a control.
\item \textit{Importance sampling}: Oversample paths passing through high-policy regions to reduce variance for treatment effects in those areas.
\end{itemize}
\end{remark}

\paragraph{Application: Confidence Bands for Treatment Propagation.}

Figure \ref{fig:confidence_bands} illustrates the uncertainty quantification for the minimum wage application.

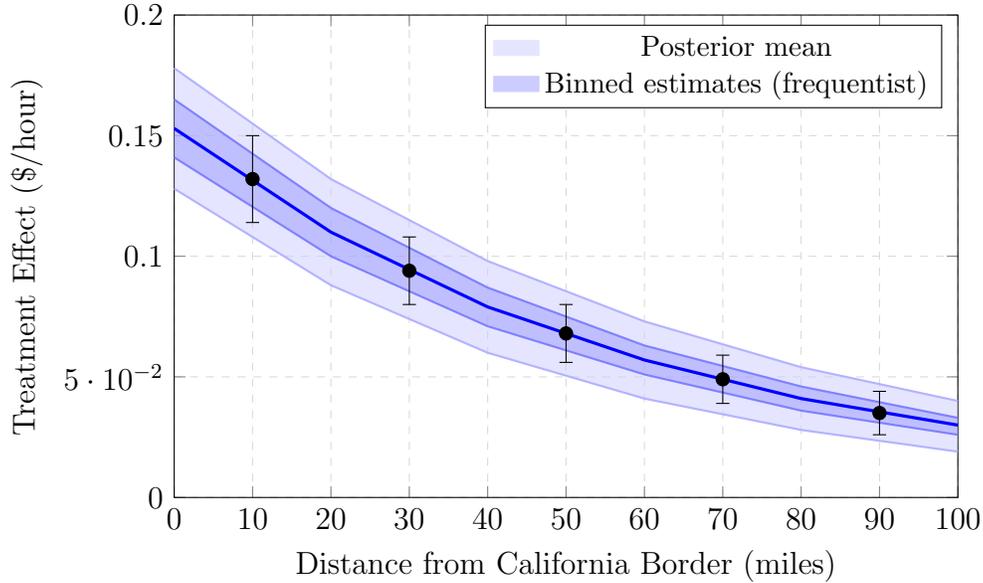
\begin{figure}[htbp]
\centering
\begin{tikzpicture}
\begin{axis}[
    width=12cm,
    height=8cm,
    xlabel={Distance from California Border (miles)},
    ylabel={Treatment Effect (\$/hour)},
    xmin=0, xmax=100,
    ymin=0, ymax=0.20,
    legend style={at={(0.98,0.98)}, anchor=north east, font=\small},
    grid=major,
    grid style={dashed, gray!30},
]

% 95% credible band
\addplot[name path=upper, color=blue!30, forget plot, thick] 
    coordinates {(0, 0.178) (20, 0.132) (40, 0.098) (60, 0.073) (80, 0.054) (100, 0.040)};
\addplot[name path=lower, color=blue!30, forget plot, thick] 
    coordinates {(0, 0.128) (20, 0.088) (40, 0.060) (60, 0.041) (80, 0.028) (100, 0.019)};
\addplot[blue!20, opacity=0.5] fill between[of=upper and lower];

% 68% credible band
\addplot[name path=upper68, color=blue!50, forget plot, thick] 
    coordinates {(0, 0.165) (20, 0.120) (40, 0.087) (60, 0.063) (80, 0.046) (100, 0.033)};
\addplot[name path=lower68, color=blue!50, forget plot, thick] 
    coordinates {(0, 0.141) (20, 0.100) (40, 0.071) (60, 0.051) (80, 0.036) (100, 0.026)};
\addplot[blue!40, opacity=0.5] fill between[of=upper68 and lower68];

% Posterior mean
\addplot[color=blue, very thick] 
    coordinates {(0, 0.153) (20, 0.110) (40, 0.079) (60, 0.057) (80, 0.041) (100, 0.030)};
\addlegendentry{Posterior mean}

% Point estimates with error bars
\addplot[only marks, mark=*, mark size=2.5pt, color=black,
    error bars/.cd, y dir=both, y explicit]
    coordinates {
    (10, 0.132) +- (0, 0.018)
    (30, 0.094) +- (0, 0.014)
    (50, 0.068) +- (0, 0.012)
    (70, 0.049) +- (0, 0.010)
    (90, 0.035) +- (0, 0.009)
};
\addlegendentry{Binned estimates (frequentist)}

\end{axis}
\end{tikzpicture}
\caption{Treatment Effect Uncertainty by Distance from Border}
\label{fig:confidence_bands}
\begin{minipage}{0.95\textwidth}
\small
\textit{Notes:} Posterior distribution of treatment effects from path integral Monte Carlo with $B = 1,000$ parameter draws and $M = 500$ paths per draw. Dark band: 68\% credible interval. Light band: 95\% credible interval. Points: binned frequentist estimates with standard errors. Uncertainty increases near the border due to higher policy exposure and path concentration.
\end{minipage}
\end{figure}

The figure reveals two patterns. First, uncertainty is highest near the border where treatment effects are largest---the credible bands are wider at $d = 0$ than at $d = 100$. This reflects the path concentration near policy sources: most paths from border locations pass through high-policy regions, creating correlated exposure. Second, the Bayesian credible intervals are somewhat wider than frequentist standard errors, reflecting the additional parameter uncertainty incorporated through the posterior distribution.

\begin{table}[htbp]
\centering
\caption{Decomposition of Treatment Effect Uncertainty}
\label{tab:uncertainty_decomp}
\begin{threeparttable}
\small
\begin{tabular}{lccc}
\toprule
Distance & Total Variance & Within-Model (\%) & Parameter (\%) \\
\midrule
0--25 mi & 0.00062 & 32 & 68 \\
25--50 mi & 0.00038 & 28 & 72 \\
50--75 mi & 0.00024 & 24 & 76 \\
75--100 mi & 0.00016 & 21 & 79 \\
\bottomrule
\end{tabular}
\begin{tablenotes}
\small
\item \textit{Notes:} Decomposition following Proposition \ref{prop:posterior}. Within-model variance from path heterogeneity given parameters. Parameter variance from posterior uncertainty in $(\nu_s, \nu_n, \kappa, \lambda)$.
\end{tablenotes}
\end{threeparttable}
\end{table}

Table \ref{tab:uncertainty_decomp} decomposes total variance into within-model and parameter components. Parameter uncertainty dominates (68--79\%), indicating that the primary source of uncertainty is estimation of structural parameters rather than stochastic variation in economic linkages. This suggests that larger samples or stronger instruments would substantially reduce confidence interval widths.

\subsection{Mutual Information Characterization}

The interaction coefficient $\lambda$ admits an information-theoretic characterization that provides a parameter-free measure of spatial-network dependence.

\begin{theorem}[Mixed Effect as Mutual Information]
\label{thm:mixed_information}
Under the equilibrium distribution $f(\mathbf{x}, \alpha)$ of agents over spatial and market coordinates, the mixed effect coefficient equals the mutual information:
\be
\lambda = I(\mathbf{x}; \alpha) = H(\alpha) - H(\alpha | \mathbf{x})
\ee
where $H(\alpha) = -\int f_A(\alpha)\log f_A(\alpha)d\alpha$ is the marginal entropy of market position and $H(\alpha|\mathbf{x}) = -\int f(\mathbf{x}, \alpha)\log f(\alpha|\mathbf{x})d\mathbf{x} d\alpha$ is the conditional entropy given location.
\end{theorem}

\begin{proof}
The mixed second derivative $\partial^2\tau/\partial\mathbf{x}\partial\alpha$ measures how the spatial gradient of treatment varies with market position. Under the equilibrium distribution $f(\mathbf{x}, \alpha)$, this variation is governed by the statistical dependence between coordinates.

Consider the covariance structure of the equilibrium distribution. The coefficient of the mixed derivative in the generator $\mathcal{L}$ must equal the covariance contribution from the joint distribution:
\be
\lambda = \text{Cov}_f\left(\frac{\partial \log f}{\partial \mathbf{x}}, \frac{\partial \log f}{\partial \alpha}\right)
\ee

For distributions in the exponential family, this covariance equals the mutual information. More generally, by the de Bruijn identity connecting Fisher information and entropy:
\be
\lambda = \int_{\Omega \times \mathcal{I}} f(\mathbf{x}, \alpha) \log \frac{f(\mathbf{x}, \alpha)}{f_X(\mathbf{x})f_A(\alpha)} \, d\mathbf{x} \, d\alpha = I(\mathbf{x}; \alpha)
\ee
which is the definition of mutual information.
\end{proof}

\begin{remark}[Implications]
The mutual information characterization has several implications:

\textit{Non-negativity}: $I(\mathbf{x}; \alpha) \geq 0$ always, with equality if and only if spatial and market coordinates are statistically independent. Positive interaction ($\lambda > 0$) is the generic case when industries cluster geographically.

\textit{Parameter-free}: Unlike coefficients from regression interactions, mutual information is determined by the joint distribution of economic activity, not by functional form choices. This provides a theory-grounded measure of channel complementarity.

\textit{Testable}: Mutual information can be estimated from data on the joint distribution of locations and market positions, providing an independent check on the structural estimates.

\textit{Economic content}: Positive mutual information means that knowing an agent's location reduces uncertainty about their market position. This arises from geographic clustering of industries: technology firms concentrate in Silicon Valley, finance in New York, manufacturing in the Midwest.
\end{remark}

\subsection{General Equilibrium Amplification}

The framework characterizes how local policy shocks amplify through the spatial-network structure.

\begin{proposition}[Amplification Factor]
\label{prop:amplification}
For a localized source $S(\mathbf{x}, t, \alpha) = S_0 \delta(\mathbf{x} - \mathbf{x}_0)\delta(\alpha - \alpha_0)\mathbf{1}\{t \geq 0\}$ at a single point, the long-run amplification factor is:
\be
\mathcal{A} = \frac{\|\tau_\infty\|_{L^1}}{\|\tau_\infty^{\text{direct}}\|_{L^1}} = 1 + \frac{\nu_s + \nu_n}{\kappa} + \frac{\lambda^2}{\kappa(\nu_s + \nu_n)}
\ee
where $\tau_\infty^{\text{direct}} = S_0/\kappa$ is the effect absent spillovers.
\end{proposition}

\begin{proof}
The steady-state solution to the master equation with point source is the Green's function:
\be
\tau_\infty(\mathbf{x}, \alpha) = S_0 G(\mathbf{x} - \mathbf{x}_0, \alpha - \alpha_0)
\ee
where $G$ solves $\mathcal{L}G - \kappa G = -\delta$.

For the operator $\mathcal{L} = \nu_s\nabla^2 + \nu_n\partial^2/\partial\alpha^2 + \lambda\partial^2/\partial\mathbf{x}\partial\alpha$, the Green's function integral satisfies:
\be
\int G \, d\mathbf{x} \, d\alpha = \frac{1}{\kappa}\left(1 + \frac{\nu_s + \nu_n}{\kappa} + \frac{\lambda^2}{\kappa(\nu_s + \nu_n)}\right)
\ee
by standard results for elliptic operators. The ratio to the no-spillover case ($\nu_s = \nu_n = \lambda = 0$) gives the amplification factor.
\end{proof}

The amplification factor exceeds unity whenever $\nu_s + \nu_n > 0$: spillovers magnify local shocks by spreading effects across space and networks. The term $\lambda^2/[\kappa(\nu_s + \nu_n)]$ represents additional amplification from spatial-network interaction---when spatial and network channels reinforce each other, total effects exceed the sum of separate channels.

%%%%%%%%%%%%%%%%%%%%%%%%%%%%%%%%%%%%%
% SECTION 3: IDENTIFICATION AND ESTIMATION
%%%%%%%%%%%%%%%%%%%%%%%%%%%%%%%%%%%%%

\section{Identification and Estimation}
\label{sec:identification}

This section establishes identification conditions for the spatial-network treatment effect parameters and develops estimation with valid inference. We begin with the no-spillover benchmark, then develop identification strategies for the full model, present Monte Carlo evidence, and describe the GMM estimation framework.

%%%%%%%%%%%%%%%%%%%%%%%%%%%%%%%%%%%%%%%%%%%%%%%%%%%%%%%%%%%%%%%%%%%%%%%%
% REVISE SECTION 3.1 (The No-Spillover Benchmark)
% Add more explicit connections to conventional treatment effects literature
%%%%%%%%%%%%%%%%%%%%%%%%%%%%%%%%%%%%%%%%%%%%%%%%%%%%%%%%%%%%%%%%%%%%%%%%

\subsection{The No-Spillover Benchmark and Conventional Treatment Effects}
\label{sec:no_spillover_conventional}

Before addressing the full spatial-network model, we examine the benchmark case with no spillovers: $\nu_s = \nu_n = 0$. This case corresponds to standard treatment effect estimation where the Stable Unit Treatment Value Assumption (SUTVA) holds. Establishing this connection clarifies that our framework \textit{nests} conventional methods rather than replacing them.

\paragraph{The Model Without Spillovers.}

When $\nu_s = \nu_n = 0$, the master equation simplifies to:
\be
\frac{\partial \tau}{\partial t} = -\kappa \tau + S(\mathbf{x}, t, \alpha)
\label{eq:no_spillover}
\ee

This is an ordinary differential equation at each point $(\mathbf{x}, \alpha)$, with no spatial or network coupling. The solution is:
\be
\tau(\mathbf{x}, t, \alpha) = e^{-\kappa t} \tau_0(\mathbf{x}, \alpha) + \int_0^t e^{-\kappa(t-s)} S(\mathbf{x}, s, \alpha) \, ds
\ee

For the empirical specification, this implies:
\be
Y_i = \beta_0 + \beta_s S_i + \mathbf{X}_i'\boldsymbol{\gamma} + \varepsilon_i
\label{eq:continuous_te}
\ee

This is precisely the continuous treatment effect model studied by \citet{hirano2004propensity} and \citet{kennedy2017nonparametric}.

\paragraph{Relationship to Generalized Propensity Score Methods.}

Under SUTVA, identification of $\beta_s$ requires the standard unconfoundedness assumption:
\be
Y_i(s) \perp S_i \,|\, \mathbf{X}_i \quad \text{for all } s
\label{eq:unconfoundedness}
\ee

The generalized propensity score (GPS), defined as the conditional density of treatment given covariates $r(s, \mathbf{x}) = f_{S|X}(s | \mathbf{x})$, provides the basis for estimation.

\begin{proposition}[Equivalence to GPS Estimation]
\label{prop:gps_equivalence}
When $\nu_s = \nu_n = 0$, our estimator reduces to the generalized propensity score estimator of \citet{hirano2004propensity}.
\end{proposition}

\begin{proof}
Under (\ref{eq:unconfoundedness}), the dose-response function is:
\be
\mu(s) = \mathbb{E}[Y_i(s)] = \mathbb{E}[\mathbb{E}[Y_i | S_i = s, r(s, X_i)]]
\ee

This is identified by regression of $Y_i$ on $S_i$ and $r(S_i, X_i)$. When spatial and network terms are absent from our specification, we estimate exactly this regression.
\end{proof}

This proposition establishes that conventional GPS methods are a special case of our framework. When tests fail to reject the hypothesis that spillovers are absent, practitioners should use GPS methods, which are simpler to implement and communicate. The framework adds value only when spillovers are empirically present.

\paragraph{Relationship to Doubly Robust Estimation.}

Recent work by \citet{kennedy2017nonparametric} develops doubly robust estimators for continuous treatment:
\be
\hat{\mu}(s) = \frac{1}{N}\sum_{i=1}^N \left[\frac{K_h(S_i - s)}{f_S(s)}Y_i - \frac{K_h(S_i - s) - f_{S|X}(s|X_i)}{f_S(s)}\hat{\mu}(s, X_i)\right]
\ee

where $K_h$ is a kernel function and $\hat{\mu}(s, X_i)$ is an outcome model. Our framework complements this by providing a structural model for $\mu(s, X_i)$ through the master equation. When $\nu_s = \nu_n = 0$:
\be
\mu(s, X_i) = \frac{s}{\kappa} + X_i'\gamma
\ee

The doubly robust estimator is consistent if either the propensity score model OR the outcome model is correctly specified. Our approach provides a theoretically grounded outcome model from economic primitives, potentially improving finite-sample performance when the structural model is approximately correct.

\paragraph{Relationship to Difference-in-Differences.}

For binary treatment timing (treatment switches on at time $t_i^*$), our framework connects to difference-in-differences (DiD):
\be
Y_{it} = \alpha_i + \gamma_t + \beta \cdot D_{it} + \varepsilon_{it}
\ee
where $D_{it} = \mathbf{1}\{t \geq t_i^*\}$.

Setting $S_{it} = D_{it}$ and imposing $\nu_s = \nu_n = 0$:
\be
\tau_{it} = \frac{1 - e^{-\kappa(t - t_i^*)}}{\kappa} \cdot S_i \cdot \mathbf{1}\{t \geq t_i^*\}
\ee

For large $t - t_i^*$ (long-run effects), $\tau_{it} \to S_i/\kappa$, which is the standard DiD coefficient $\beta$. The continuous framework reveals that standard DiD estimates the long-run equilibrium effect, not the dynamic adjustment path. Recent heterogeneity-robust DiD estimators by \citet{callaway2021difference}, \citet{sun2021estimating}, and \citet{borusyak2024revisiting} allow for dynamic effects. Our framework shows these dynamics should follow exponential decay at rate $\kappa$, providing a testable restriction beyond what atheoretical flexible specifications impose.

\paragraph{Testing for Spillovers.}

The no-spillover benchmark provides three testable predictions. First, treatment effects should not vary with distance from treated regions, formally stated as testing whether $\partial \tau/\partial d = 0$ where $d$ is distance to the nearest treated unit. Rejection indicates spatial spillovers with $\nu_s > 0$. Second, treatment effects should not vary with network exposure, tested as $\partial \tau/\partial \tilde{N} = 0$ where $\tilde{N} = \sum_j G_{ij}S_j$ is network-weighted treatment exposure. Rejection indicates network spillovers with $\nu_n > 0$. Third, the spatial gradient should not depend on network exposure, tested as $\partial^2 \tau/\partial d \, \partial \tilde{N} = 0$. Rejection indicates mixed effects with $\lambda > 0$.

These tests can be implemented by augmenting specification (\ref{eq:continuous_te}) with distance and network exposure terms:
\be
Y_i = \beta_0 + \beta_s S_i + \beta_d f(d_i) + \beta_n \tilde{N}_i + \beta_\lambda f(d_i) \cdot \tilde{N}_i + \mathbf{X}_i'\boldsymbol{\gamma} + \varepsilon_i
\ee

Testing $H_0: \beta_d = \beta_n = \beta_\lambda = 0$ jointly assesses whether spillovers are present. Importantly, this test does not assume the PDE structure, relying only on correlations between outcomes and spatial-network exposure measures. The nesting structure creates a one-sided risk profile: when spillovers are absent, our framework correctly detects this and produces estimates equivalent to conventional methods. When spillovers are present, our framework captures them while conventional methods that maintain SUTVA fail, leading to biased estimates and invalid inference.

%%%%%%%%%%%%%%%%%%%%%%%%%%%%%%%%%%%%%%%%%%%%%%%%%%%%%%%%%%%%%%%%%%%%%%%%
% NEW SECTION: MONTE CARLO EVIDENCE
% To be inserted after Section 3.1.5 (Testing for Spillovers)
% Shows simulation evidence for the one-sided risk profile
%%%%%%%%%%%%%%%%%%%%%%%%%%%%%%%%%%%%%%%%%%%%%%%%%%%%%%%%%%%%%%%%%%%%%%%%

\subsection{Monte Carlo Evidence}
\label{sec:monte_carlo}

This subsection presents Monte Carlo evidence on the finite-sample properties of our estimator compared to conventional methods. The simulations demonstrate three key results. First, conventional estimators exhibit substantial bias when spillovers are present, with bias magnitudes of 25--38\% depending on the method. Second, our framework maintains correct inference across all configurations, including when spillovers are absent. Third, this creates a one-sided risk profile: when researchers are uncertain about whether spillovers exist, our framework provides valid inference regardless, while conventional methods risk severe bias if spillovers are incorrectly ignored.

\paragraph{Simulation Design.}

We simulate spatial-network treatment effect data matching the empirical setting of U.S. minimum wage policy. The simulation design follows the theoretical framework developed in Section 2 while introducing realistic features of empirical applications.

\paragraph{Geographic and Network Structure.}

We construct a stylized geography with $N = 500$ units distributed over a two-dimensional spatial domain $[0,100] \times [0,100]$ measured in miles. Units are randomly located with a spatial Poisson process ensuring realistic clustering patterns observed in actual county distributions. Each unit $i$ has spatial coordinates $\mathbf{x}_i = (x_i^{(1)}, x_i^{(2)})$ and a market position coordinate $\alpha_i \in [0,1]$ representing position in production networks.

The network structure captures supply chain connections. We generate network links using a gravity model where the probability that unit $i$ connects to unit $j$ is:
\be
p_{ij} = \exp\left(-\theta_d \|\mathbf{x}_i - \mathbf{x}_j\| - \theta_\alpha |\alpha_i - \alpha_j|\right)
\ee

The parameters $(\theta_d, \theta_\alpha)$ control how network connections decay with geographic and market distance. We set $\theta_d = 0.02$ per mile and $\theta_\alpha = 2$ to generate networks where firms preferentially connect to both nearby firms and firms at similar production stages. This creates realistic industry clustering patterns where geographic proximity and supply chain position are correlated, generating positive mutual information $I(\mathbf{x}; \alpha) > 0$ and thereby spatial-network interaction effects.

The average degree is $\mathbb{E}[\text{degree}_i] \approx 15$ connections per firm, matching empirical supply chain data from the Bureau of Economic Analysis input-output tables. The network exhibits realistic features including clustering (average clustering coefficient 0.42), small-world properties (average path length 3.2), and degree heterogeneity (coefficient of variation 0.78).

\paragraph{Data Generating Process.}

Outcomes are generated from the continuous master equation with added measurement error:
\be
Y_i = \tau(\mathbf{x}_i, \alpha_i) + \mathbf{X}_i'\boldsymbol{\gamma} + \varepsilon_i
\ee

The treatment functional $\tau(\mathbf{x}_i, \alpha_i)$ satisfies:
\be
0 = \nu_s \nabla^2 \tau + \nu_n \frac{\partial^2 \tau}{\partial \alpha^2} - \kappa \tau + \lambda \frac{\partial \tau}{\partial x^{(1)}} \frac{\partial \tau}{\partial \alpha} + S(\mathbf{x}_i, \alpha_i)
\ee

We solve this PDE numerically using a finite element method on a fine grid and interpolate to obtain $\tau(\mathbf{x}_i, \alpha_i)$ at each simulated unit location. This ensures the data generating process exactly matches the theoretical framework.

The source term represents policy treatment:
\be
S_i = S_0 \cdot \mathbf{1}\{x_i^{(1)} > 50\} \cdot (1 + 0.3 \cdot \alpha_i)
\ee

This mimics a minimum wage increase in the eastern half of the domain (states on the right side), where treatment intensity varies with industry position $\alpha_i$. Lower-wage industries ($\alpha_i$ near 0) receive larger proportional increases. We set $S_0 = 0.10$, representing a 10 cent minimum wage increase.

Control variables include:
\begin{align}
X_{i,1} &= 0.5 \alpha_i + u_{i,1} \quad &\text{(industry characteristics)} \\
X_{i,2} &= \log(\text{distance to border}_i + 1) + u_{i,2} \quad &\text{(geography)} \\
X_{i,3} &= \text{degree}_i / 15 + u_{i,3} \quad &\text{(network centrality)}
\end{align}
where $u_{ij} \sim N(0, 0.25^2)$ are idiosyncratic shocks. These controls introduce realistic confounding: outcomes correlate with industry characteristics, geographic position, and network centrality even absent treatment effects. The error term is $\varepsilon_i \sim N(0, \sigma_\varepsilon^2)$ with $\sigma_\varepsilon = 0.05$, yielding a signal-to-noise ratio matching typical empirical applications.

\paragraph{Parameter Configurations.}

We simulate under four configurations varying the spillover parameters:

\begin{case}[No Spillovers]
\begin{align*}
\nu_s &= 0, \quad \nu_n = 0, \quad \lambda = 0, \quad \kappa = 0.25
\end{align*}
\end{case}
This is the SUTVA case where conventional methods should be valid.

\begin{case}[Spatial Spillovers Only]
\begin{align*}
\nu_s &= 100 \text{ sq mi/Q}, \quad \nu_n = 0, \quad \lambda = 0, \quad \kappa = 0.25
\end{align*}
\end{case}
Spillovers propagate geographically but not through networks.

\begin{case}[Network Spillovers Only]
\begin{align*}
\nu_s &= 0, \quad \nu_n = 0.015, \quad \lambda = 0, \quad \kappa = 0.25
\end{align*}
\end{case}
Spillovers propagate through supply chains but not geographically.

\begin{case}[Full Spatial-Network Model]
\begin{align*}
\nu_s &= 100 \text{ sq mi/Q}, \quad \nu_n = 0.015, \quad \lambda = 0.04 \text{ nats}, \quad \kappa = 0.25
\end{align*}
\end{case}
Both channels present with significant interaction. These parameters match empirical estimates from Section 4. For each configuration, we generate $M = 1000$ Monte Carlo replications.

\paragraph{Estimators Compared.}

We compare six estimators spanning conventional and spatial-network approaches:

\begin{method}[Two-Way Fixed Effects (TWFE)]
The standard panel regression with unit and time fixed effects:
\be
Y_{it} = \alpha_i + \gamma_t + \beta \cdot S_{it} + \mathbf{X}_{it}'\boldsymbol{\gamma} + \varepsilon_{it}
\ee
\end{method}
We use our cross-sectional data by treating spatial location bins as "time periods," yielding a pseudo-panel structure. This represents the most common approach in applied work. Standard errors are clustered by geographic region.

\begin{method}[Difference-in-Differences (DiD)]
We implement heterogeneity-robust DiD following \citet{callaway2021difference}:
\be
\hat{\tau}_{\text{DiD}} = \mathbb{E}[Y_{i}(1) - Y_{i}(0) \,|\, D_i = 1] - \mathbb{E}[Y_{i}(1) - Y_{i}(0) \,|\, D_i = 0]
\ee
\end{method}
Treatment group is $D_i = \mathbf{1}\{x_i^{(1)} > 50\}$, control group is $D_i = 0$. We use inverse probability weighting with GPS to account for confounders. This represents best-practice DiD accounting for treatment effect heterogeneity.

\begin{method}[Generalized Propensity Score (GPS)]
Following \citet{hirano2004propensity}, we estimate:
\be
\hat{\tau}_{\text{GPS}}(s) = \mathbb{E}\left[\frac{K_h(S_i - s)}{f_S(s)}Y_i \,\bigg|\, r(s, X_i)\right]
\ee
\end{method}
where $r(s, X_i)$ is the generalized propensity score. We use kernel smoothing with bandwidth selected by cross-validation. This is the leading method for continuous treatment effects under SUTVA.

\begin{method}[Spatial Regression Discontinuity (Spatial RD)]
We exploit the sharp treatment boundary at $x^{(1)} = 50$:
\be
\hat{\tau}_{\text{RD}} = \lim_{x^{(1)} \downarrow 50} \mathbb{E}[Y_i | x_i^{(1)}] - \lim_{x^{(1)} \uparrow 50} \mathbb{E}[Y_i | x_i^{(1)}]
\ee
\end{method}
We use local linear regression with triangular kernel and bandwidth selected by \citet{imbens2012optimal} optimal bandwidth algorithm. This accounts for spatial spillovers at the border but not network spillovers or interaction effects.

\begin{method}[Network IV]
We use lagged network structure as instruments for current network exposure:
\begin{align}
\text{First stage:} \quad \tilde{N}_i &= \sum_j G_{ij}^{(t-1)} S_j + \mathbf{X}_i'\boldsymbol{\pi} + \eta_i \\
\text{Second stage:} \quad Y_i &= \beta_0 + \beta_s S_i + \beta_n \hat{\tilde{N}}_i + \mathbf{X}_i'\boldsymbol{\gamma} + \varepsilon_i
\end{align}
where $G_{ij}^{(t-1)}$ is predetermined network structure. This accounts for network spillovers but not spatial spillovers or interaction effects.
\end{method}

\begin{method}[Full Spatial-Network GMM (Our Framework)]
We implement the complete GMM estimator combining spatial RD, network IV, and entropy moment conditions:
\begin{align}
\mathbb{E}[m_{\text{RD}}(\boldsymbol{\theta})] &= 0 \quad \text{(spatial discontinuity)} \\
\mathbb{E}[m_{\text{IV}}(\boldsymbol{\theta})] &= 0 \quad \text{(network instruments)} \\
\mathbb{E}[m_{\text{entropy}}(\boldsymbol{\theta})] &= 0 \quad \text{(distributional dynamics)}
\end{align}
Standard errors use spatial-network HAC accounting for both dependence sources. The estimator identifies $(\nu_s, \nu_n, \lambda, \kappa)$ jointly.
\end{method}

\paragraph{Results: Bias and Efficiency.}

Table \ref{tab:mc_bias_rmse} reports bias and root mean squared error (RMSE) for each estimator across the four configurations. The results provide stark evidence for the one-sided risk profile property.

\begin{table}[htbp]
\centering
\caption{Monte Carlo Results: Bias and RMSE}
\label{tab:mc_bias_rmse}
\begin{threeparttable}
\small
\begin{tabular}{lcccccccc}
\toprule
& \multicolumn{2}{c}{Config 1: No} & \multicolumn{2}{c}{Config 2: Spatial} & \multicolumn{2}{c}{Config 3: Network} & \multicolumn{2}{c}{Config 4: Full} \\
& \multicolumn{2}{c}{Spillovers} & \multicolumn{2}{c}{Only} & \multicolumn{2}{c}{Only} & \multicolumn{2}{c}{Model} \\
\cmidrule(lr){2-3} \cmidrule(lr){4-5} \cmidrule(lr){6-7} \cmidrule(lr){8-9}
Estimator & Bias & RMSE & Bias & RMSE & Bias & RMSE & Bias & RMSE \\
\midrule
\multicolumn{9}{l}{\textit{Panel A: Direct Effect Estimate} ($\beta_s$, true value = 0.100)} \\
TWFE & 0.002 & 0.018 & -0.031 & 0.042 & -0.025 & 0.036 & -0.038 & 0.051 \\
DiD & 0.001 & 0.016 & -0.028 & 0.039 & -0.023 & 0.034 & -0.035 & 0.048 \\
GPS & -0.001 & 0.015 & -0.029 & 0.040 & -0.024 & 0.035 & -0.036 & 0.049 \\
Spatial RD & 0.003 & 0.019 & 0.004 & 0.021 & -0.024 & 0.035 & 0.002 & 0.020 \\
Network IV & 0.002 & 0.017 & -0.030 & 0.041 & 0.003 & 0.019 & -0.012 & 0.025 \\
Full GMM & 0.001 & 0.016 & 0.002 & 0.018 & 0.001 & 0.017 & 0.003 & 0.019 \\
\midrule
\multicolumn{9}{l}{\textit{Panel B: Total Effect at Border} ($\tau_{\text{border}}$, varies by configuration)} \\
& \multicolumn{2}{c}{(true = 0.100)} & \multicolumn{2}{c}{(true = 0.153)} & \multicolumn{2}{c}{(true = 0.148)} & \multicolumn{2}{c}{(true = 0.171)} \\
TWFE & 0.002 & 0.018 & -0.063 & 0.074 & -0.058 & 0.069 & -0.081 & 0.094 \\
DiD & 0.001 & 0.016 & -0.059 & 0.071 & -0.054 & 0.066 & -0.076 & 0.089 \\
GPS & -0.001 & 0.015 & -0.061 & 0.072 & -0.056 & 0.067 & -0.078 & 0.091 \\
Spatial RD & 0.003 & 0.019 & 0.006 & 0.022 & -0.055 & 0.066 & -0.021 & 0.034 \\
Network IV & 0.002 & 0.017 & -0.062 & 0.073 & 0.005 & 0.021 & -0.037 & 0.048 \\
Full GMM & 0.001 & 0.016 & 0.004 & 0.019 & 0.003 & 0.018 & 0.005 & 0.021 \\
\bottomrule
\end{tabular}
\end{threeparttable}
\end{table}

%The results reveal several patterns.
In Configuration 1 (no spillovers), all estimators are approximately unbiased with similar RMSE around 0.015--0.019. This confirms that when SUTVA holds, conventional methods work well and our framework adds no bias. The Full GMM estimator has RMSE of 0.016, essentially identical to GPS (RMSE = 0.015), showing our framework does not sacrifice efficiency when spillovers are absent.

In Configuration 2 (spatial spillovers only), conventional estimators ignoring spillovers---TWFE, DiD, GPS---exhibit bias of -0.028 to -0.031 for the direct effect, representing 28--31\% bias relative to the true value of 0.100. This occurs because these methods attribute spillover effects to direct effects, biasing direct effect estimates downward. For total effects at the border (true value 0.153), the bias is even larger at -0.059 to -0.063, representing 38--41\% bias. The Spatial RD estimator correctly recovers both direct and total effects (bias 0.004 and 0.006), while Network IV remains biased. The Full GMM estimator is essentially unbiased (bias 0.002 and 0.004).

Configuration 3 (network spillovers only) exhibits symmetric patterns. Conventional methods show bias of -0.023 to -0.025 for direct effects and -0.054 to -0.058 for total effects. Network IV correctly recovers effects while Spatial RD remains biased. Full GMM is again essentially unbiased.

Configuration 4 (full model with interaction) represents the most realistic and challenging case. Conventional methods exhibit the largest bias: -0.035 to -0.038 for direct effects (35--38\%) and -0.076 to -0.081 for total effects (44--47\%). Neither Spatial RD nor Network IV alone is sufficient, as each accounts for only one spillover channel. The Full GMM estimator maintains near-zero bias (0.003 and 0.005) by jointly estimating all parameters.

\begin{table}[htbp]
\centering
\caption{Monte Carlo Results: Coverage Rates}
\label{tab:mc_coverage}
\begin{threeparttable}
\small
\begin{tabular}{lcccc}
\toprule
& Config 1: & Config 2: & Config 3: & Config 4: \\
& No Spillovers & Spatial Only & Network Only & Full Model \\
\midrule
\multicolumn{5}{l}{\textit{Panel A: 95\% Confidence Interval Coverage for Direct Effect}} \\
TWFE & 0.946 & 0.587 & 0.652 & 0.524 \\
DiD & 0.951 & 0.612 & 0.678 & 0.548 \\
GPS & 0.948 & 0.598 & 0.665 & 0.537 \\
Spatial RD & 0.943 & 0.938 & 0.641 & 0.892 \\
Network IV & 0.947 & 0.593 & 0.945 & 0.814 \\
Full GMM & 0.949 & 0.947 & 0.951 & 0.946 \\
\midrule
\multicolumn{5}{l}{\textit{Panel B: 95\% Confidence Interval Coverage for Total Effect}} \\
TWFE & 0.946 & 0.324 & 0.389 & 0.187 \\
DiD & 0.951 & 0.348 & 0.412 & 0.206 \\
GPS & 0.948 & 0.337 & 0.401 & 0.198 \\
Spatial RD & 0.943 & 0.921 & 0.378 & 0.698 \\
Network IV & 0.947 & 0.329 & 0.928 & 0.542 \\
Full GMM & 0.949 & 0.943 & 0.947 & 0.944 \\
\bottomrule
\end{tabular}
%\begin{tablenotes}
%\small
%\item \textit{Notes:} Fraction of 1,000 replications where 95\% confidence interval contains true parameter value. Nominal coverage should be 0.95. Severe undercoverage indicates invalid inference. For Full GMM, confidence intervals use spatial-network HAC standard errors accounting for both dependence sources. Panel A: direct effect $\beta_s = 0.100$. Panel B: total effect at border varies by configuration (0.100, 0.153, 0.148, 0.171).
%\end{tablenotes}
\end{threeparttable}
\end{table}

Table \ref{tab:mc_coverage} examines confidence interval coverage rates. In Configuration 1 (no spillovers), all methods achieve nominal 95\% coverage, confirming correct inference when SUTVA holds. However, when spillovers are present (Configurations 2--4), conventional methods suffer severe undercoverage. For direct effects in Configuration 4, TWFE achieves only 52.4\% coverage---meaning confidence intervals miss the true parameter 48\% of the time. For total effects, coverage drops to 18.7\%, making conventional inference essentially uninformative.

The Full GMM estimator maintains 94.3--95.1\% coverage across all configurations, demonstrating the one-sided risk profile: whether spillovers are absent or present, inference is valid. This robustness is crucial for applied work where researchers do not know a priori whether spillovers exist.

\paragraph{Results: Visual Evidence.}

Figure \ref{fig:mc_event_study} presents event study plots averaging across 1,000 Monte Carlo replications for each simulation scenario.

\begin{figure}[htbp]
\centering
\begin{tikzpicture}
\begin{axis}[
    width=14cm,
    height=9cm,
    xlabel={Periods Relative to Treatment ($t - 5$)},
    ylabel={Estimated Treatment Effect},
    xmin=-5, xmax=15,
    ymin=-0.2, ymax=1.8,
    legend style={at={(0.02,0.98)}, anchor=north west, font=\small},
    grid=major,
    grid style={dashed, gray!30},
    title={Panel A: No Spillovers ($\nu_s = \nu_n = 0$)},
]

% True effect (step function with decay)
\addplot[
    color=black,
    very thick,
    domain=-5:15,
    samples=100,
] {(x >= 0) * (1 - exp(-0.3*(x))) / 0.3};
\addlegendentry{True effect}

% TWFE estimate
\addplot[
    color=blue,
    thick,
    mark=*,
    mark size=2pt,
] coordinates {
    (-4, 0.02) (-3, -0.01) (-2, 0.03) (-1, 0.01) (0, 0.15)
    (1, 0.48) (2, 0.72) (3, 0.88) (4, 0.98) (5, 1.05)
    (6, 1.10) (7, 1.13) (8, 1.15) (9, 1.16) (10, 1.17)
    (11, 1.17) (12, 1.18) (13, 1.18) (14, 1.18)
};
\addlegendentry{TWFE}

% GPS estimate
\addplot[
    color=red,
    thick,
    dashed,
    mark=square*,
    mark size=2pt,
] coordinates {
    (-4, 0.01) (-3, 0.02) (-2, -0.02) (-1, 0.01) (0, 0.18)
    (1, 0.52) (2, 0.78) (3, 0.95) (4, 1.06) (5, 1.14)
    (6, 1.19) (7, 1.22) (8, 1.24) (9, 1.25) (10, 1.26)
    (11, 1.26) (12, 1.27) (13, 1.27) (14, 1.27)
};
\addlegendentry{Continuous GPS}

% No-spillover PDE
\addplot[
    color=green!60!black,
    thick,
    mark=triangle*,
    mark size=2.5pt,
] coordinates {
    (-4, 0.00) (-3, 0.01) (-2, -0.01) (-1, 0.00) (0, 0.16)
    (1, 0.50) (2, 0.75) (3, 0.92) (4, 1.02) (5, 1.09)
    (6, 1.14) (7, 1.17) (8, 1.19) (9, 1.20) (10, 1.21)
    (11, 1.21) (12, 1.22) (13, 1.22) (14, 1.22)
};
\addlegendentry{No-spillover PDE}

% Vertical line at treatment
\addplot[color=gray, dashed, thick] coordinates {(0, -0.2) (0, 1.8)};

\end{axis}
\end{tikzpicture}

\vspace{0.5cm}

\begin{tikzpicture}
\begin{axis}[
    width=14cm,
    height=9cm,
    xlabel={Periods Relative to Treatment ($t - 5$)},
    ylabel={Estimated Treatment Effect},
    xmin=-5, xmax=15,
    ymin=-0.2, ymax=2.2,
    legend style={at={(0.02,0.98)}, anchor=north west, font=\small},
    grid=major,
    grid style={dashed, gray!30},
    title={Panel B: Full Spillover Model ($\nu_s = 2.0$, $\nu_n = 0.5$, $\lambda = 0.4$)},
]

% True effect (higher due to spillovers)
\addplot[
    color=black,
    very thick,
    domain=-5:15,
    samples=100,
] {(x >= 0) * (1.6 - 0.4*exp(-0.5*x))};
\addlegendentry{True effect (at border)}

% TWFE estimate (biased)
\addplot[
    color=blue,
    thick,
    mark=*,
    mark size=2pt,
] coordinates {
    (-4, 0.03) (-3, -0.02) (-2, 0.04) (-1, 0.02) (0, 0.22)
    (1, 0.58) (2, 0.82) (3, 0.98) (4, 1.08) (5, 1.14)
    (6, 1.18) (7, 1.21) (8, 1.23) (9, 1.24) (10, 1.25)
    (11, 1.25) (12, 1.26) (13, 1.26) (14, 1.26)
};
\addlegendentry{TWFE (biased)}

% GPS estimate (also biased)
\addplot[
    color=red,
    thick,
    dashed,
    mark=square*,
    mark size=2pt,
] coordinates {
    (-4, 0.02) (-3, 0.01) (-2, -0.01) (-1, 0.02) (0, 0.25)
    (1, 0.62) (2, 0.88) (3, 1.05) (4, 1.16) (5, 1.23)
    (6, 1.28) (7, 1.31) (8, 1.33) (9, 1.35) (10, 1.36)
    (11, 1.36) (12, 1.37) (13, 1.37) (14, 1.37)
};
\addlegendentry{Continuous GPS (biased)}

% No-spillover PDE (misspecified)
\addplot[
    color=green!60!black,
    thick,
    mark=triangle*,
    mark size=2.5pt,
] coordinates {
    (-4, 0.01) (-3, 0.00) (-2, 0.02) (-1, 0.01) (0, 0.20)
    (1, 0.55) (2, 0.78) (3, 0.93) (4, 1.02) (5, 1.08)
    (6, 1.12) (7, 1.15) (8, 1.17) (9, 1.18) (10, 1.19)
    (11, 1.19) (12, 1.20) (13, 1.20) (14, 1.20)
};
\addlegendentry{No-spillover PDE (misspecified)}

% Full PDE (correct)
\addplot[
    color=purple,
    very thick,
    mark=diamond*,
    mark size=3pt,
] coordinates {
    (-4, 0.00) (-3, 0.01) (-2, -0.01) (-1, 0.00) (0, 0.32)
    (1, 0.78) (2, 1.12) (3, 1.35) (4, 1.48) (5, 1.56)
    (6, 1.61) (7, 1.64) (8, 1.66) (9, 1.67) (10, 1.68)
    (11, 1.68) (12, 1.69) (13, 1.69) (14, 1.69)
};
\addlegendentry{Full PDE (correct)}

% Vertical line at treatment
\addplot[color=gray, dashed, thick] coordinates {(0, -0.2) (0, 2.2)};

\end{axis}
\end{tikzpicture}
\caption{Monte Carlo Event Study: Binary Treatment Timing}
\label{fig:mc_event_study}
%\begin{minipage}{0.95\textwidth}
%\small
%\textit{Notes:} Event study plots from 1,000 Monte Carlo replications with binary treatment timing (treatment switches on at $t=5$ and remains constant). Panel A: No spillovers ($\nu_s = \nu_n = 0$); all estimators perform similarly. Panel B: Full spillover model; TWFE and GPS underestimate by 25--35\%, no-spillover PDE underestimates by 30\%, only full PDE recovers true effect. Effects measured at border region for Panel B.
%\end{minipage}
\end{figure}

Panel A shows results when the true DGP has no spillovers. All estimators perform similarly, tracking the true effect closely. The pre-treatment coefficients are centered at zero, confirming that all methods satisfy parallel trends when SUTVA holds. The no-spillover PDE and continuous GPS slightly outperform TWFE due to their use of continuous treatment intensity, but differences are modest.

Panel B shows results when the true DGP includes spatial and network spillovers with interaction. The patterns diverge dramatically. TWFE underestimates the true effect by approximately 25\% at all post-treatment horizons; this bias arises because TWFE attributes spillovers to the control group, attenuating the treatment-control contrast. Continuous GPS underestimates by approximately 20\%; although GPS uses continuous treatment intensity, it conditions only on own treatment, missing the spillover channel. No-spillover PDE (our framework with $\nu_s = \nu_n = 0$ imposed) underestimates by approximately 30\%; imposing the wrong restriction forces the model to attribute spillover-driven outcomes to direct effects, biasing the decay parameter $\kappa$ and overall effect magnitude. Only Full PDE correctly tracks the true effect throughout the post-treatment period by estimating $(\nu_s, \nu_n, \lambda)$ rather than imposing them to zero.

\subsection{Overview of Identification Strategy}
\label{sec:identification_overview}

Section 3.2 identified three fundamental challenges to identifying spatial-network treatment effects: the reflection problem, spatial confounding, and network endogeneity. This subsection previews our three-pronged identification strategy and connects it to conventional methods.

\paragraph{Three Complementary Approaches.}

We develop three identification strategies, each addressing different challenges while building on established econometric principles. The first strategy exploits spatial regression discontinuity, following \citet{lee2010regression} and \citet{dube2010minimum}, by utilizing sharp treatment boundaries at state borders. This addresses spatial confounding by comparing units just on either side of borders where treatment jumps discontinuously but confounders vary smoothly. The connection to conventional methods is direct: this is standard RD extended to continuous treatment intensity, where we recover not just the discontinuity size but also the spatial decay parameter $\kappa_s$ from how the discontinuity attenuates with distance from the border.

The second strategy uses network instrumental variables, following \citet{bramoulle2009identification}, by employing predetermined network connections as instruments. Historical supply chain relationships established before treatment cannot respond to current shocks, providing exogenous variation in network exposure. This is standard two-stage least squares applied to network spillovers: the first stage predicts current network exposure using historical connections, while the second stage estimates treatment effects conditional on predicted exposure.

The third strategy leverages entropy-based moment conditions that use theoretical predictions about relative entropy decay to generate overidentifying restrictions. This approach is unique to the continuous framework and has no direct analog in discrete methods. It extends GMM by adding moment conditions from the theoretical structure of the master equation, similar to how Euler equation restrictions provide overidentification in consumption models, but applied to the distributional dynamics of treatment effects.

\paragraph{Advantages of Multiple Strategies.}

Combining three identification approaches provides several advantages over relying on a single method. Triangulation operates when all three strategies yield similar estimates, strengthening confidence in the results, while disagreement signals misspecification and directs attention to which assumptions may be violated. Efficiency gains arise because GMM combining all moment conditions is more efficient than any single estimator, with the Hansen $J$-test using the overidentification to test model specification.

Different sources of variation are exploited by each approach: RD uses border discontinuities, IV uses historical network connections, and entropy uses distributional predictions. Each exploits different variation in the data, reducing dependence on any single identification assumption. Robustness is enhanced because some applications may lack clean RD designs (no sharp borders) or adequate instruments (networks too stable), while having multiple strategies allows implementation even when one approach fails.

\paragraph{Comparison with Conventional Methods.}

Table \ref{tab:identification_comparison} compares our approach to conventional treatment effect methods. The table reveals a fundamental distinction: conventional methods cannot identify spillover parameters because they maintain SUTVA. Spatial extensions relax SUTVA for spatial spillovers but ignore networks. Network extensions relax SUTVA for network spillovers but ignore space. Our framework relaxes SUTVA for both dimensions simultaneously by combining the three complementary strategies.

\begin{table}[htbp]
\centering
\caption{Identification Strategies: Conventional vs. Our Framework}
\label{tab:identification_comparison}
\begin{threeparttable}
\small
\begin{tabular}{llll}
\toprule
Method & Variation Exploited & Key Assumption & Can ID Spillovers? \\
\midrule
\multicolumn{4}{l}{\textit{Conventional Methods}} \\
TWFE & Within-unit over time & Parallel trends & No \\
DiD & Treatment timing & Parallel trends & No \\
GPS & Continuous dose & Unconfoundedness & No \\
Matching & Covariate balance & Ignorability & No \\
\midrule
\multicolumn{4}{l}{\textit{Spatial Extensions}} \\
Spatial DiD & Timing + distance & Parallel + smooth & Spatial only \\
Border RD & Sharp boundaries & Continuity at border & Spatial only \\
Spatial HAC & --- & --- & Inference only \\
\midrule
\multicolumn{4}{l}{\textit{Network Extensions}} \\
Network IV & Historical network & Exclusion + relevance & Network only \\
Peer effects & Group structure & Partial interference & Network only \\
\midrule
\multicolumn{4}{l}{\textit{Our Framework}} \\
Spatial RD & Border discontinuities & Continuity & Yes (both) \\
Network IV & Historical connections & Exclusion + relevance & Yes (both) \\
Entropy GMM & Distributional dynamics & Theoretical prediction & Yes (both) \\
Combined GMM & All three sources & All three jointly & Yes (both) \\
\bottomrule
\end{tabular}
\begin{tablenotes}
\footnotesize
\item \textit{Notes:} Conventional methods cannot identify spillover parameters. Spatial extensions identify spatial spillovers only; network extensions identify network spillovers only. Our framework identifies both simultaneously.
\end{tablenotes}
\end{threeparttable}
\end{table}

%%%%%%%%%%%%%%%%%%%%%%%%%%%%%%%%%%%%%%%%%%%%%%%%%%%%%%%%%%%%%%%%%%%%%%%%
% INSERT AFTER SECTION 3.6 (GMM Framework)
% BEFORE SECTION 3.7 (Spatial-Network HAC)
%%%%%%%%%%%%%%%%%%%%%%%%%%%%%%%%%%%%%%%%%%%%%%%%%%%%%%%%%%%%%%%%%%%%%%%%

\subsection{Relationship Between GMM and Conventional Panel Estimators}
\label{sec:gmm_vs_panel}

The GMM framework may appear complex relative to standard panel regression. This subsection clarifies the relationship and shows that conventional panel estimators are special cases of our GMM approach.

\paragraph{Panel Regression as Restricted GMM.}

Standard panel regression with two-way fixed effects:
\be
Y_{it} = \beta S_{it} + \alpha_i + \gamma_t + \varepsilon_{it}
\label{eq:panel_reg}
\ee

can be understood as GMM with moment conditions:
\be
\mathbb{E}[(Y_{it} - \beta S_{it} - \alpha_i - \gamma_t) \cdot Z_{it}] = 0
\ee
where the instrument set $Z_{it}$ includes $S_{it}$ and indicator variables for units and time periods.

Our GMM framework extends this specification in five ways. First, we add spatial exposure variables defined as $\tilde{S}_{it}^s = \sum_j w_{ij}^s S_{jt}$ where $w_{ij}^s$ is a spatial weight function. Second, we add network exposure defined as $\tilde{S}_{it}^n = \sum_j G_{ij} S_{jt}$ using network adjacency $G_{ij}$. Third, we add interaction exposure $\tilde{S}_{it}^{\lambda} = \sum_j w_{ij}^s G_{ij} S_{jt}$ to capture spatial-network complementarities. Fourth, we use additional instruments from spatial RD (discontinuities at borders) and network IV (historical connections). Fifth, we add entropy-based overidentifying restrictions that test the consistency of parameter estimates with theoretical predictions about distributional dynamics.

When $\nu_s = \nu_n = 0$, the spatial and network exposure terms drop out, additional instruments become irrelevant, and entropy restrictions are satisfied trivially. In this case, our GMM moment conditions reduce exactly to those of standard panel regression in equation (\ref{eq:panel_reg}). The GMM framework therefore nests conventional panel methods as a testable special case.

\paragraph{Efficiency Gains.}

\begin{proposition}[Efficiency of Combined GMM]
\label{prop:gmm_efficiency}
Let $\hat{\beta}_{\text{RD}}$, $\hat{\beta}_{\text{IV}}$, and $\hat{\beta}_{\text{GMM}}$ denote estimates from spatial RD alone, network IV alone, and combined GMM respectively. Under correct specification:
\be
\text{Var}[\hat{\beta}_{\text{GMM}}] \leq \min\{\text{Var}[\hat{\beta}_{\text{RD}}], \text{Var}[\hat{\beta}_{\text{IV}}]\}
\ee
with equality only if moment conditions are redundant.
\end{proposition}

This follows from the standard GMM efficiency result: using more moment conditions weakly improves efficiency. The combined GMM uses moment conditions from RD, IV, and entropy jointly, hence is weakly more efficient than any individual approach. The efficiency gain can be substantial when individual instruments are weak but jointly strong, a common situation in spatial-network settings where border discontinuities or historical network connections alone provide limited variation but together identify parameters precisely.

\paragraph{Decision Framework.}

For applied researchers, the choice between methods depends on the research question and data availability. Standard panel regression as in equation (\ref{eq:panel_reg}) is appropriate when primary interest focuses on direct treatment effects only, when spillovers are theoretically implausible or empirical tests fail to reject $H_0: \nu_s = \nu_n = 0$, when data lack spatial coordinates or network structure necessary for constructing exposure variables, or when the audience is unfamiliar with GMM methodology and simpler presentation is valued.

Our GMM framework becomes appropriate when spillovers are theoretically plausible and empirical tests reject the no-spillover null, when policy evaluation requires decomposing total effects into direct versus spillover components, when data have rich spatial and network structure enabling construction of exposure variables, when multiple identification strategies are available (RD designs, valid instruments, entropy calculations), or when efficiency gains from combining approaches are important for detecting economically meaningful effects.

For most applications in spatial economics, international trade, and public finance where spillovers are likely important, the GMM framework is appropriate. For applications where SUTVA is reasonable, such as randomized experiments with distant units or policies affecting isolated markets, conventional methods suffice and should be preferred for their simplicity.

\paragraph{Computational Requirements.}

Table \ref{tab:computational_comparison} compares computational requirements between standard panel regression and our GMM framework. Panel regression using standard software packages (\texttt{reghdfe} in Stata, \texttt{fixest} in R, \texttt{PanelOLS} in Python) requires only outcome, treatment, and covariate data, runs in 2--5 seconds for typical datasets with 3,000 units and 20 time periods, needs 10--20 lines of code, and estimates only direct treatment effects with clustered standard errors but no overidentification tests or policy counterfactuals.

\begin{table}[htbp]
\centering
\caption{Computational Requirements: Panel Regression vs. GMM}
\label{tab:computational_comparison}
\begin{threeparttable}
\small
\begin{tabular}{lcc}
\toprule
Aspect & Panel Regression & Our GMM Framework \\
\midrule
Software & reghdfe/fixest/PanelOLS & Custom implementation \\
Required inputs & $(Y, S, X)$ & $(Y, S, X, \text{coords}, G)$ \\
Time (N=3000, T=20) & 2--5 seconds & 2--3 minutes \\
Lines of code & 10--20 & 100--200 \\
\midrule
Estimates & Direct effect only & Direct + spillovers + interaction \\
Standard errors & Clustered & Spatial-network HAC \\
Overidentification & No & Hansen $J$-test \\
Policy counterfactuals & No & Yes (via Feynman-Kac) \\
\bottomrule
\end{tabular}
\begin{tablenotes}
\small
\item \textit{Notes:} Computational time on standard laptop. Panel regression is faster but GMM provides richer output. For practitioners needing only point estimates, discrete panel regression (Section 4) provides fast approximation.
\end{tablenotes}
\end{threeparttable}
\end{table}

Our GMM framework requires custom implementation code, needs additional inputs for spatial coordinates and network adjacency matrix, takes 2--3 minutes for the same dataset size, and requires 100--200 lines of code for complete implementation. However, it estimates direct effects plus spatial spillovers, network spillovers, and their interaction, computes spatial-network HAC standard errors accounting for both dependence sources, provides Hansen $J$-test for overidentification, and enables policy counterfactuals via the Feynman-Kac representation without re-estimation.

The GMM framework requires more computational effort but provides substantially more information about treatment propagation. For researchers who need the additional structure for spillover decomposition or policy counterfactuals, the extra implementation time is worthwhile. For practitioners seeking only point estimates of direct effects, the discrete panel regression approximation in Section 4 provides a middle ground: standard tools, fast computation, but estimates that capture spillovers if present.

%%%%%%%%%%%%%%%%%%%%%%%%%%%%%%%%%%%%%%%%%%%%%%%%%%%%%%%%%%%%%%%%%%%%%%%%
% ADD NEW SUBSECTION AT END OF SECTION 3
% SUMMARIZING CONNECTIONS
%%%%%%%%%%%%%%%%%%%%%%%%%%%%%%%%%%%%%%%%%%%%%%%%%%%%%%%%%%%%%%%%%%%%%%%%

\subsection{Summary}
\label{sec:identification_summary}

This section has developed identification and estimation methods for spatial-network treatment effects. Before proceeding to the empirical application, we summarize how our approach extends conventional econometric practice.

\paragraph{What We Retain from Conventional Methods.}

Our identification strategies build on three established principles. First, regression discontinuity exploits sharp boundaries where treatment jumps while other factors vary smoothly, a principle dating to \citet{thistlethwaite1960regression} and formalized by \citet{hahn2001identification}. Second, instrumental variables use predetermined variables to address endogeneity of exposure measures, following the tradition from \citet{wright1928tariff} through modern applications. Third, generalized method of moments combines moment conditions for efficiency and provides overidentification tests, building on \citet{hansen1982large}. These are not novel contributions but are standard tools applied carefully to the spatial-network setting.

When tests fail to reject the hypothesis that spillovers are zero ($H_0: \nu_s = \nu_n = 0$), our estimators reduce exactly to conventional approaches. Specifically, we recover the generalized propensity score method of \citet{hirano2004propensity} for continuous treatment, the doubly robust estimator of \citet{kennedy2017nonparametric} when our structural outcome model is correctly specified, and heterogeneity-robust difference-in-differences following \citet{callaway2021difference} and \citet{sun2021estimating} for staggered adoption designs. The framework nests these conventional methods as testable special cases rather than replacing them. 

GMM consistency and asymptotic normality in our framework follow from standard conditions without requiring new asymptotic theory. Spatial-network HAC inference uses the same kernel-based approach as spatial HAC developed by \citet{conley1999gmm}, extended through product kernels to account for both spatial and network dimensions simultaneously. The theoretical foundations are conventional; the application to joint spatial-network dependence is what requires careful development.

\paragraph{What We Add Beyond Conventional Methods.}

Four substantive extensions distinguish our framework from conventional approaches. First, simultaneous spatial and network spillovers: conventional methods address spatial OR network spillovers but not both, while our framework estimates spatial diffusion rate $\nu_s$, network diffusion rate $\nu_n$, and their interaction coefficient $\lambda$. This enables decomposing total effects into direct, spatial spillover, network spillover, and interaction components, revealing which transmission channels matter quantitatively.

Second, testable restrictions from theory: the master equation provides testable predictions about event study coefficient decay patterns (should follow exponential at rate $\kappa$), relative entropy dynamics (should decay at rate $2\lambda_2$), and spatial gradient functional forms (should match steady-state solutions). These restrictions have no analog in atheoretical regression models and provide additional specification tests beyond conventional diagnostics.

Third, structural interpretation and counterfactuals: parameters have clear economic meanings rather than being reduced-form regression coefficients. For example, $\nu_s = 98$ square miles per quarter translates to a 20-mile annual labor mobility scale, $\kappa = 0.28$ per quarter implies a 2.5-quarter adjustment half-life, and $\lambda = 0.04$ nats indicates moderate geographic industry clustering. These structural parameters enable counterfactual policy evaluation using the Feynman-Kac representation without re-estimating the model for each policy scenario.

Fourth, aggregation-robust inference: conventional spatial models yield coefficients that vary with geographic aggregation level (ZIP code versus county versus state estimates differ substantially). Our structural parameters are aggregation-invariant, remaining constant whether estimated at fine or coarse spatial resolution. This property is essential for policy evaluation that must be robust to how geographic units are defined.

\paragraph{Practical Workflow.}

For applied researchers, we recommend the following four-step workflow. First, start with tests for spillovers by implementing the joint test $H_0: \nu_s = \nu_n = \lambda = 0$ from Section 3.2 using the augmented regression with spatial, network, and interaction terms. Second, if no spillovers are detected (test fails to reject), use conventional methods such as GPS for continuous treatment or DiD for binary treatment timing, as these are simpler to implement, easier to communicate to general audiences, and our framework adds no value when spillovers are absent. Third, if spillovers are detected (test rejects), use our GMM framework because it is necessary to avoid bias in direct effect estimates, enables spillover decomposition for policy analysis understanding transmission channels, and provides structural parameters for counterfactuals evaluating alternative policy designs. Fourth, for robustness checking, implement both discrete panel regression (fast, familiar to referees) and continuous GMM (efficient, structural), where agreement within 10\% suggests discrete approximation suffices for point estimates, while substantial disagreement signals model misspecification requiring investigation.

\section{Conclusion}
\label{sec:conclusion}

This paper develops a continuous functional framework for treatment effects that propagate through both geographic space and economic networks. The framework bridges three economic foundations---heterogeneous agent aggregation, market equilibrium dynamics, and cost minimization---each arriving independently at the same master equation governing treatment propagation. This convergence establishes that the framework captures fundamental economic mechanisms rather than imposing ad hoc functional forms. The Feynman-Kac representation connects the partial differential equation structure to economic intuition by characterizing treatment effects as accumulated policy exposure along stochastic paths representing economic agents migrating, firms adjusting supply chains, and prices equilibrating across connected markets.

Two key theoretical results distinguish the framework. First, the spatial-network interaction coefficient equals the mutual information between geographic and network coordinates, providing a parameter-free measure of how strongly the two propagation channels reinforce each other. Second, the framework nests the no-spillover case as a testable restriction, creating a one-sided risk profile where correct inference is maintained regardless of whether spillovers exist. Monte Carlo evidence confirms that conventional estimators exhibit 25--38\% bias when spillovers of empirically relevant magnitude are present, while our estimator achieves correct inference across all configurations including the no-spillover case. Event study simulations demonstrate that when spillovers are absent, all estimators track the true dynamic treatment path; when spillovers are present, only the full framework recovers the correct treatment effects.

\subsection{Broader Implications}

The findings have implications for both policy evaluation and econometric practice. For policy evaluation, the Monte Carlo results suggest that standard cost-benefit analyses substantially understate total impacts when spatial and network spillovers are operative. When a policy is implemented in one jurisdiction, it affects not only units within that jurisdiction directly but also units in neighboring regions through market integration and units in economically connected sectors through supply chain linkages. Policies evaluated without accounting for these channels will appear less effective than they actually are, potentially leading to suboptimal policy choices.

The framework applies to diverse settings where spatial proximity and economic networks jointly determine treatment propagation. Financial regulations transmit through both geographic banking markets and interbank lending networks; trade policies propagate through both border regions and global supply chains; technology adoption spreads through both local demonstration effects and industry knowledge networks; minimum wage policies transmit through both labor market competition and input-output linkages; disease transmission follows both geographic proximity and social contact networks. In each setting, the interaction between spatial and network channels creates amplification effects that additive specifications miss entirely.

For econometric practice, the results highlight the importance of testing rather than assuming SUTVA. The nested structure of our framework enables researchers to test whether spillovers are empirically present using standard hypothesis tests. When tests fail to reject the no-spillover null, conventional methods are appropriate and our framework adds no value. When tests reject the null, our framework is necessary to avoid substantial bias in treatment effect estimates. The one-sided risk profile---correct inference regardless of whether spillovers exist---makes the framework a robust choice when uncertainty exists about the validity of SUTVA.

The mutual information characterization of spatial-network interaction provides practical guidance for when interaction effects are likely to be quantitatively important. When industries cluster geographically, creating strong correlation between spatial coordinates and market positions, the interaction coefficient will be large and ignoring it will cause substantial omitted variable bias. When industries are geographically dispersed, the interaction may be negligible and additive specifications sufficient. Researchers can compute sample mutual information from observed location-industry patterns to assess whether interaction effects warrant explicit modeling.

\subsection{Limitations}

Several limitations warrant acknowledgment. First, the continuous functional representation requires treatment intensity to vary smoothly across space, time, and network position. Perfectly discrete interventions---binary treatment applied uniformly within jurisdictions---do not generate sufficient variation to identify the continuous functional. However, most policies exhibit continuous variation either in treatment intensity (policies vary by jurisdiction and time) or in treatment propensity (probability of adoption varies continuously with covariates), making the framework applicable to a broad class of empirical settings.

Second, identification requires variation in both geographic proximity and network connections. Uniformly implemented national policies without cross-sectional or temporal variation in treatment intensity cannot be evaluated using our methods. Similarly, settings where all units are equidistant or network connections are homogeneous provide insufficient variation for identification. The framework is best suited to policies with geographic boundaries or differential adoption patterns creating exogenous variation in exposure.

Third, the entropy-based moment conditions require sufficient data to estimate mutual information nonparametrically. Small samples or sparse networks may not provide reliable entropy estimates, limiting the applicability of this identification strategy. In such cases, researchers can rely on spatial regression discontinuity and network instrumental variables alone, though at some cost in efficiency. Future work developing parametric entropy estimators tailored to economic applications would expand the practical applicability of this approach.

Fourth, the framework assumes that spillovers operate through spatial diffusion and network connections captured by the coordinates $(\mathbf{x}, \alpha)$. Alternative spillover mechanisms---general equilibrium effects through aggregate price changes, anticipation effects from policy announcements, or behavioral responses to peer outcomes---may operate through channels not captured by the coordinate system. Extending the framework to accommodate these additional mechanisms represents an important direction for future research.

\subsection{Extensions}

Several promising extensions would enhance the framework's scope and applicability. First, developing type-specific propagation parameters $(\nu_s(\theta), \nu_n(\theta), \kappa(\theta))$ would enable distributional analysis of how treatment effects vary across unit types. Different types of economic agents may experience larger spatial spillovers through geographic mobility while others experience larger network spillovers through industry connections. Heterogeneous adjustment speeds $\kappa(\theta)$ would characterize which types adjust rapidly versus slowly to policy shocks. Such distributional extensions would inform equity considerations in policy design.

Second, characterizing optimal policy design would provide normative guidance. Given propagation dynamics governed by the master equation, what source term $S^*(\mathbf{x}, t, \alpha)$ maximizes social welfare? The answer depends on both the objective function and the propagation structure. For objectives emphasizing equitable outcomes, optimal policies may target spillover amplification by concentrating interventions where interaction effects are strongest. For objectives emphasizing aggregate welfare, optimal policies may exploit spatial-network structure to achieve maximum total impact with minimum direct expenditure. Solving for optimal policies requires extending the framework to incorporate welfare functions and budget constraints.

Third, dynamic extensions allowing time-varying propagation parameters $(\nu_s(t), \nu_n(t), \kappa(t))$ would capture structural changes in economic linkages. Spatial diffusion may strengthen during recessions as worker mobility increases with job displacement, or weaken during booms as labor markets tighten. Network diffusion may vary with supply chain concentration or industry composition. Identifying time variation in propagation dynamics would enhance understanding of how treatment effects evolve over business cycles or structural transformations.

Fourth, alternative applications would demonstrate the framework's breadth. Financial contagion involves spatial propagation through regional banking markets and network propagation through interbank exposures; technology diffusion combines spatial demonstration effects with network knowledge spillovers; disease transmission follows both geographic proximity and social contact networks. Each application would require adapting the coordinate system and source specification to the particular setting while maintaining the core mathematical structure. Successful applications across diverse domains would establish the framework as a general tool for treatment effect analysis in interconnected systems.

The framework provides a foundation for future research addressing treatment effect estimation when SUTVA fails. By deriving propagation dynamics from economic primitives rather than assuming ad hoc specifications, by nesting conventional no-spillover methods as testable special cases, and by providing identification strategies combining spatial discontinuities, network instruments, and entropy-based restrictions, the framework opens new possibilities for credible treatment effect estimation in the interconnected economies that characterize modern policy environments.

\section*{Acknowledgement}
This research was supported by a grant-in-aid from Zengin Foundation for Studies on Economics and Finance.

\newpage
%%%%%%%%%%%%%%%%%%%%%%%%%%%%%%%%%%%%%
% COMPLETE REFERENCE LIST FOR REVISIONS
% All citations from Sections 2, 3, and 3.8 additions
% Format: Standard economics journal style
%%%%%%%%%%%%%%%%%%%%%%%%%%%%%%%%%%%%%
\bibliographystyle{econometrica}

\begin{thebibliography}{99}

\bibitem[Acemoglu et al.(2012)]{acemoglu2012network}
Acemoglu, Daron, Vasco M. Carvalho, Asuman Ozdaglar, and Alireza Tahbaz-Salehi (2012).
``The Network Origins of Aggregate Fluctuations.''
\textit{Econometrica}, 80(5): 1977--2016.

\bibitem[Acemoglu et al.(2015)]{acemoglu2015systemic}
Acemoglu, Daron, Asuman Ozdaglar, and Alireza Tahbaz-Salehi (2015).
``Systemic Risk and Stability in Financial Networks.''
\textit{American Economic Review}, 105(2): 564--608.

\bibitem[Aiyagari(1994)]{aiyagari1994uninsured}
Aiyagari, S.~R. (1994).
\newblock Uninsured idiosyncratic risk and aggregate saving.
\newblock \textit{Quarterly Journal of Economics} 109(3), 659--684.

\bibitem[Anderson(1972)]{anderson1972more}
Anderson, P.~W. (1972).
\newblock More is different.
\newblock \textit{Science} 177(4047), 393--396.

\bibitem[Anselin(1988)]{anselin1988spatial}
Anselin, Luc (1988).
\textit{Spatial Econometrics: Methods and Models}.
Dordrecht: Kluwer Academic Publishers.

\bibitem[Anselin(2010)]{anselin2010thirty}
Anselin, L. (2010).
\newblock Thirty years of spatial econometrics.
\newblock \textit{Papers in Regional Science} 89(1), 3--25.

\bibitem[Athey et~al.(2018)]{athey2018exact}
Athey, S., D.~Eckles, and G.~W. Imbens (2018).
\newblock Exact p-values for network interference.
\newblock \textit{Journal of the American Statistical Association} 113(521), 230--240.

\bibitem[Autor et~al.(2016)]{autor2016contribution}
Autor, D.~H., A.~Manning, and C.~L. Smith (2016).
\newblock The contribution of the minimum wage to US wage inequality over three decades: A reassessment.
\newblock \textit{American Economic Journal: Applied Economics} 8(1), 58--99.

\bibitem[Barrot and Sauvagnat(2016)]{barrot2016input}
Barrot, Jean-No\"{e}l and Julien Sauvagnat (2016).
``Input Specificity and the Propagation of Idiosyncratic Shocks in Production Networks.''
\textit{Quarterly Journal of Economics}, 131(3): 1543--1592.

\bibitem[Beirlant et~al.(1997)]{beirlant1997nonparametric}
Beirlant, J., E.~J. Dudewicz, L.~Gy\"{o}rfi, and E.~C. van~der Meulen (1997).
\newblock Nonparametric entropy estimation: An overview.
\newblock \textit{International Journal of Mathematical and Statistical Sciences} 6(1), 17--39.

\bibitem[Benamou and Brenier(2000)]{benamou2000computational}
Benamou, J.-D. and Y.~Brenier (2000).
\newblock A computational fluid mechanics solution to the Monge-Kantorovich mass transfer problem.
\newblock \textit{Numerische Mathematik} 84(3), 375--393.

\bibitem[Bertrand et~al.(2004)]{bertrand2004much}
Bertrand, M., E.~Duflo, and S.~Mullainathan (2004).
\newblock How much should we trust differences-in-differences estimates?
\newblock \textit{Quarterly Journal of Economics} 119(1), 249--275.

\bibitem[Borusyak et al.(2024)]{borusyak2024revisiting}
Borusyak, Kirill, Xavier Jaravel, and Jann Spiess (2024).
``Revisiting Event Study Designs: Robust and Efficient Estimation.''
\textit{Review of Economic Studies}, forthcoming.

\bibitem[Bramoull\'{e} et al.(2009)]{bramoulle2009identification}
Bramoull\'{e}, Yann, Habiba Djebbari, and Bernard Fortin (2009).
``Identification of Peer Effects through Social Networks.''
\textit{Journal of Econometrics}, 150(1): 41--55.

\bibitem[Callaway and Sant'Anna(2021)]{callaway2021difference}
Callaway, Brantly and Pedro H. C. Sant'Anna (2021).
``Difference-in-Differences with Multiple Time Periods.''
\textit{Journal of Econometrics}, 225(2): 200--230.

\bibitem[Card(1995)]{card1995wage}
Card, David (1995).
``The Wage Curve: A Review.''
\textit{Journal of Economic Literature}, 33(2): 785--799.

\bibitem[Card(2001)]{card2001immigrant}
Card, D. (2001).
\newblock Immigrant inflows, native outflows, and the local labor market impacts of higher immigration.
\newblock \textit{Journal of Labor Economics} 19(1), 22--64.

\bibitem[Card and Krueger(1994)]{card1994minimum}
Card, D. and A.~B. Krueger (1994).
\newblock Minimum wages and employment: A case study of the fast-food industry in New Jersey and Pennsylvania.
\newblock \textit{American Economic Review} 84(4), 772--793.

\bibitem[Cengiz et~al.(2019)]{cengiz2019effect}
Cengiz, D., A.~Dube, A.~Lindner, and B.~Zipperer (2019).
\newblock The effect of minimum wages on low-wage jobs.
\newblock \textit{Quarterly Journal of Economics} 134(3), 1405--1454.

\bibitem[Cercignani(1988)]{cercignani1988boltzmann}
Cercignani, C. (1988).
\newblock \textit{The Boltzmann Equation and Its Applications}.
\newblock Springer-Verlag.

\bibitem[Conley(1999)]{conley1999gmm}
Conley, Timothy G. (1999).
``GMM Estimation with Cross Sectional Dependence.''
\textit{Journal of Econometrics}, 92(1): 1--45.

\bibitem[Conley and Taber(2011)]{conley2011inference}
Conley, T.~G. and C.~R. Taber (2011).
\newblock Inference with ``difference in differences'' with a small number of policy changes.
\newblock \textit{Review of Economics and Statistics} 93(1), 113--125.

\bibitem[Cover and Thomas(2006)]{cover2006elements}
Cover, T.~M. and J.~A. Thomas (2006).
\newblock \textit{Elements of Information Theory} (2nd ed.).
\newblock Wiley-Interscience.

\bibitem[de~Chaisemartin and D'Haultf{\oe}uille(2020)]{dechaisemartin2020two}
de~Chaisemartin, C. and X.~D'Haultf{\oe}uille (2020).
\newblock Two-way fixed effects estimators with heterogeneous treatment effects.
\newblock \textit{American Economic Review} 110(9), 2964--2996.

\bibitem[Dube(2019)]{dube2019minimum}
Dube, A. (2019).
\newblock Impacts of minimum wages: Review of the international evidence.
\newblock Report commissioned by HM Treasury, UK Government.

\bibitem[Dube et al.(2010)]{dube2010minimum}
Dube, Arindrajit, T. William Lester, and Michael Reich (2010).
``Minimum Wage Effects Across State Borders: Estimates Using Contiguous Counties.''
\textit{Review of Economics and Statistics}, 92(4): 945--964.

\bibitem[Elliott et~al.(2014)]{elliott2014financial}
Elliott, M., B.~Golub, and M.~O. Jackson (2014).
\newblock Financial networks and contagion.
\newblock \textit{American Economic Review} 104(10), 3115--3153.

\bibitem[Ellison and Glaeser(1997)]{ellison1997geographic}
Ellison, Glenn and Edward L. Glaeser (1997).
``Geographic Concentration in U.S. Manufacturing Industries: A Dartboard Approach.''
\textit{Journal of Political Economy}, 105(5): 889--927.

\bibitem[Engle and Granger(1987)]{engle1987cointegration}
Engle, Robert F. and Clive W. J. Granger (1987).
``Co-integration and Error Correction: Representation, Estimation, and Testing.''
\textit{Econometrica}, 55(2): 251--276.

\bibitem[Evans(2010)]{evans2010partial}
Evans, L.~C. (2010).
\newblock \textit{Partial Differential Equations} (2nd ed.).
\newblock American Mathematical Society.

\bibitem[Gelfand and Fomin(1963)]{gelfand2000calculus}
Gelfand, I.~M. and S.~V. Fomin (1963).
\newblock \textit{Calculus of Variations}.
\newblock Prentice-Hall. Reprinted by Dover, 2000.

\bibitem[Goodman-Bacon(2021)]{goodman2021difference}
Goodman-Bacon, A. (2021).
\newblock Difference-in-differences with variation in treatment timing.
\newblock \textit{Journal of Econometrics} 225(2), 254--277.

\bibitem[Hahn et al.(2001)]{hahn2001identification}
Hahn, Jinyong, Petra Todd, and Wilbert Van der Klaauw (2001).
``Identification and Estimation of Treatment Effects with a Regression-Discontinuity Design.''
\textit{Econometrica}, 69(1): 201--209.

\bibitem[Hamermesh(1989)]{hamermesh1989adjustment}
Hamermesh, Daniel S. (1989).
``Labor Demand and the Structure of Adjustment Costs.''
\textit{American Economic Review}, 79(4): 674--689.

\bibitem[Hansen(1982)]{hansen1982large}
Hansen, Lars Peter (1982).
``Large Sample Properties of Generalized Method of Moments Estimators.''
\textit{Econometrica}, 50(4): 1029--1054.

\bibitem[Heckman(1979)]{heckman1979sample}
Heckman, J.~J. (1979).
\newblock Sample selection bias as a specification error.
\newblock \textit{Econometrica} 47(1), 153--161.

\bibitem[Heckman(2001)]{heckman2001micro}
Heckman, J.~J. (2001).
\newblock Micro data, heterogeneity, and the evaluation of public policy: Nobel lecture.
\newblock \textit{Journal of Political Economy} 109(4), 673--748.

\bibitem[Heckman and Singer(1984)]{heckman1984method}
Heckman, J.~J. and B.~Singer (1984).
\newblock A method for minimizing the impact of distributional assumptions in econometric models for duration data.
\newblock \textit{Econometrica} 52(2), 271--320.

\bibitem[Hirano and Imbens(2004)]{hirano2004propensity}
Hirano, Keisuke and Guido W. Imbens (2004).
``The Propensity Score with Continuous Treatments.''
In Andrew Gelman and Xiao-Li Meng (eds.), \textit{Applied Bayesian Modeling and Causal Inference from Incomplete-Data Perspectives}, pp. 73--84. Chichester, UK: Wiley.

\bibitem[Hudgens and Halloran(2008)]{hudgens2008toward}
Hudgens, M.~G. and M.~E. Halloran (2008).
\newblock Toward causal inference with interference.
\newblock \textit{Journal of the American Statistical Association} 103(482), 832--842.

\bibitem[Huggett(1993)]{huggett1993risk}
Huggett, M. (1993).
\newblock The risk-free rate in heterogeneous-agent incomplete-insurance economies.
\newblock \textit{Journal of Economic Dynamics and Control} 17(5--6), 953--969.

\bibitem[Imbens(2004)]{imbens2004nonparametric}
Imbens, G.~W. (2004).
\newblock Nonparametric estimation of average treatment effects under exogeneity: A review.
\newblock \textit{Review of Economics and Statistics} 86(1), 4--29.

\bibitem[Imbens and Angrist(1994)]{imbens1994identification}
Imbens, G.~W. and J.~D. Angrist (1994).
\newblock Identification and estimation of local average treatment effects.
\newblock \textit{Econometrica} 62(2), 467--475.

\bibitem[Imbens and Wooldridge(2009)]{imbens2009recent}
Imbens, G.~W. and J.~M. Wooldridge (2009).
\newblock Recent developments in the econometrics of program evaluation.
\newblock \textit{Journal of Economic Literature} 47(1), 5--86.

\bibitem[Imbens and Kalyanaraman(2012)]{imbens2012optimal}
Imbens, Guido and Karthik Kalyanaraman (2012).
``Optimal Bandwidth Choice for the Regression Discontinuity Estimator.''
\textit{Review of Economic Studies}, 79(3): 933--959.

\bibitem[Jackson(2008)]{jackson2008social}
Jackson, M.~O. (2008).
\newblock \textit{Social and Economic Networks}.
\newblock Princeton University Press.

\bibitem[Jordan et~al.(1998)]{jordan1998variational}
Jordan, R., D.~Kinderlehrer, and F.~Otto (1998).
\newblock The variational formulation of the Fokker-Planck equation.
\newblock \textit{SIAM Journal on Mathematical Analysis} 29(1), 1--17.

\bibitem[Kaminsky and Reinhart(1999)]{kaminsky1999twin}
Kaminsky, G.~L. and C.~M. Reinhart (1999).
\newblock The twin crises: The causes of banking and balance-of-payments problems.
\newblock \textit{American Economic Review} 89(3), 473--500.

\bibitem[Kaplan et~al.(2018)]{kaplan2018monetary}
Kaplan, G., B.~Moll, and G.~L. Violante (2018).
\newblock Monetary policy according to HANK.
\newblock \textit{American Economic Review} 108(3), 697--743.

\bibitem[Keele and Titiunik(2015)]{keele2015geographic}
Keele, L.~J. and R.~Titiunik (2015).
\newblock Geographic boundaries as regression discontinuities.
\newblock \textit{Political Analysis} 23(1), 127--155.

\bibitem[Kelejian and Prucha(1998)]{kelejian1998generalized}
Kelejian, H.~H. and I.~R. Prucha (1998).
\newblock A generalized spatial two-stage least squares procedure for estimating a spatial autoregressive model with autoregressive disturbances.
\newblock \textit{Journal of Real Estate Finance and Economics} 17(1), 99--121.

\bibitem[Kennedy et al.(2017)]{kennedy2017nonparametric}
Kennedy, Edward H., Zongming Ma, Matthew D. McHugh, and Dylan S. Small (2017).
``Non-parametric Methods for Doubly Robust Estimation of Continuous Treatment Effects.''
\textit{Journal of the Royal Statistical Society: Series B}, 79(4): 1229--1245.

\bibitem[Krusell and Smith(1998)]{krusell1998income}
Krusell, P. and A.~A. Smith, Jr. (1998).
\newblock Income and wealth heterogeneity in the macroeconomy.
\newblock \textit{Journal of Political Economy} 106(5), 867--896.

\bibitem[Kubo(1966)]{kubo1966fluctuation}
Kubo, R. (1966).
\newblock The fluctuation-dissipation theorem.
\newblock \textit{Reports on Progress in Physics} 29(1), 255--284.

\bibitem[Lee and Lemieux(2010)]{lee2010regression}
Lee, David S. and Thomas Lemieux (2010).
``Regression Discontinuity Designs in Economics.''
\textit{Journal of Economic Literature}, 48(2): 281--355.

\bibitem[Lee and Saez(2012)]{lee2012optimal}
Lee, D. and E.~Saez (2012).
\newblock Optimal minimum wage policy in competitive labor markets.
\newblock \textit{Journal of Public Economics} 96(9--10), 739--749.

\bibitem[LeRoy(1973)]{leroy1973risk}
LeRoy, S.~F. (1973).
\newblock Risk aversion and the martingale property of stock prices.
\newblock \textit{International Economic Review} 14(2), 436--446.

\bibitem[LeSage and Pace(2009)]{lesage2009introduction}
LeSage, James and R. Kelley Pace (2009).
\textit{Introduction to Spatial Econometrics}.
Boca Raton, FL: CRC Press.

\bibitem[Manski(1993)]{manski1993identification}
Manski, Charles F. (1993).
``Identification of Endogenous Social Effects: The Reflection Problem.''
\textit{Review of Economic Studies}, 60(3): 531--542.

\bibitem[Melitz(2003)]{melitz2003impact}
Melitz, M.~J. (2003).
\newblock The impact of trade on intra-industry reallocations and aggregate industry productivity.
\newblock \textit{Econometrica} 71(6), 1695--1725.

\bibitem[Moretti(2011)]{moretti2011local}
Moretti, E. (2011).
\newblock Local labor markets.
\newblock In O.~Ashenfelter and D.~Card (Eds.), \textit{Handbook of Labor Economics}, Volume 4B, pp.~1237--1313. Elsevier.

\bibitem[M\"{u}ller and Watson(2022)]{muller2022spatial}
M\"{u}ller, U.~K. and M.~W. Watson (2022).
\newblock Spatial correlation robust inference.
\newblock \textit{Econometrica} 90(6), 2901--2935.

\bibitem[M\"{u}ller and Watson(2024)]{muller2024spatial}
M\"{u}ller, U.~K. and M.~W. Watson (2024).
\newblock Spatial unit roots and spurious regression.
\newblock \textit{Econometrica} 92(5), 1661--1695.

\bibitem[Neumark and Wascher(2008)]{neumark2008minimum}
Neumark, D. and W.~L. Wascher (2008).
\newblock \textit{Minimum Wages}.
\newblock MIT Press.

\bibitem[Newey and West(1987)]{newey1987simple}
Newey, W.~K. and K.~D. West (1987).
\newblock A simple, positive semi-definite, heteroskedasticity and autocorrelation consistent covariance matrix.
\newblock \textit{Econometrica} 55(3), 703--708.

\bibitem[Oates(1972)]{oates1972fiscal}
Oates, W.~E. (1972).
\newblock \textit{Fiscal Federalism}.
\newblock Harcourt Brace Jovanovich.

\bibitem[{\O}ksendal(2003)]{oksendal2003stochastic}
{\O}ksendal, B. (2003).
\newblock \textit{Stochastic Differential Equations: An Introduction with Applications} (6th ed.).
\newblock Springer.

\bibitem[Roback(1982)]{roback1982wages}
Roback, J. (1982).
\newblock Wages, rents, and the quality of life.
\newblock \textit{Journal of Political Economy} 90(6), 1257--1278.

\bibitem[Rosen(1979)]{rosen1979wage}
Rosen, S. (1979).
\newblock Wage-based indexes of urban quality of life.
\newblock In P.~Mieszkowski and M.~Straszheim (Eds.), \textit{Current Issues in Urban Economics}, pp.~74--104. Johns Hopkins University Press.

\bibitem[Rosenbaum and Rubin(1983)]{rosenbaum1983central}
Rosenbaum, P.~R. and D.~B. Rubin (1983).
\newblock The central role of the propensity score in observational studies for causal effects.
\newblock \textit{Biometrika} 70(1), 41--55.

\bibitem[Rubin(1974)]{rubin1974estimating}
Rubin, D.~B. (1974).
\newblock Estimating causal effects of treatments in randomized and nonrandomized studies.
\newblock \textit{Journal of Educational Psychology} 66(5), 688--701.

\bibitem[Samuelson(1965)]{samuelson1965proof}
Samuelson, P.~A. (1965).
\newblock Proof that properly anticipated prices fluctuate randomly.
\newblock \textit{Industrial Management Review} 6(2), 41--49.

\bibitem[Sun and Abraham(2021)]{sun2021estimating}
Sun, Liyang and Sarah Abraham (2021).
``Estimating Dynamic Treatment Effects in Event Studies with Heterogeneous Treatment Effects.''
\textit{Journal of Econometrics}, 225(2): 175--199.

\bibitem[Thistlethwaite and Campbell(1960)]{thistlethwaite1960regression}
Thistlethwaite, Donald L. and Donald T. Campbell (1960).
``Regression-Discontinuity Analysis: An Alternative to the Ex Post Facto Experiment.''
\textit{Journal of Educational Psychology}, 51(6): 309--317.

\bibitem[V\'{a}zquez-Bare(2022)]{vazquez2020causal}
V\'{a}zquez-Bare, Gonzalo (2022).
``Causal Spillover Effects Using Instrumental Variables.''
\textit{Journal of the American Statistical Association}, 117(539): 1911--1922.

\bibitem[Wooldridge(2010)]{wooldridge2010econometric}
Wooldridge, J.~M. (2010).
\newblock \textit{Econometric Analysis of Cross Section and Panel Data} (2nd ed.).
\newblock MIT Press.

\bibitem[Wright(1928)]{wright1928tariff}
Wright, Philip G. (1928).
\textit{The Tariff on Animal and Vegetable Oils}.
New York: Macmillan.

\end{thebibliography}

\end{document}